\pdfoutput=1
\documentclass[final,finalrunningheads,a4paper]{llncs}
\usepackage[utf8]{inputenc}

\usepackage{breakurl}
\usepackage{eucal}

\usepackage{centernot}
\usepackage{mdframed}

\usepackage{paralist}
\usepackage{mathtools}
\usepackage{url}
\usepackage{extarrows}

\usepackage{etoolbox}
\usepackage{xspace} 

\usepackage{stmaryrd} 

\usepackage[inline,shortlabels]{enumitem} 
\newlist{inlinelist}{enumerate*}{1}
\setlist*[inlinelist,1]{%
  label=(\roman*),
}

\usepackage{xifthen}  
\usepackage{ifthen}

\usepackage{nicefrac}

\usepackage{color}

\usepackage{tikz} 
\usepackage{float}

\usepackage{stackengine} 

\usepackage{pifont}
\newcommand{\cmark}{\ding{51}}%
\newcommand{\xmark}{\ding{55}}%

\usepackage{hyperref}
\usepackage{cleveref}

\usepackage[final,nomargin,inline,index]{fixme} 
\fxusetheme{color}

\definecolor{Black}{HTML}{000000}
\definecolor{Gray}{HTML}{808080}
\definecolor{Magenta}{HTML}{FF00FF}
\definecolor{RubineRed}{HTML}{ED017D}
\definecolor{ForestGreen}{HTML}{028A0F}
\definecolor{MidnightBlue}{HTML}{006795}
\definecolor{Plum}{HTML}{92268F}

\FXRegisterAuthor{bart}{anbart}{\color{magenta} {\underline{bart}}}
\FXRegisterAuthor{zun}{anzun}{\color{OliveGreen} {\underline{zun}}}

\hypersetup{
  breaklinks   = true,
  colorlinks   = true, 
  urlcolor     = blue, 
  linkcolor    = blue, 
  citecolor    = red   
}
\hypersetup{final} 

\usepackage{listings}
\definecolor{listingBG}{HTML}{FFFFCB}%
\definecolor{listingFrame}{HTML}{BBBB98}%
\definecolor{listingLineno}{rgb}{0.5,0.5,1.0}%

\definecolor{LightGrey}{rgb}{0.975,0.975,0.975}

\lstdefinelanguage{txscript}{
	commentstyle=\color{Gray},
	morecomment=[l]{//},
	morecomment=[s]{/*}{*/},
	classoffset=0,
        escapechar=\$,
	morekeywords={if,then,else,contract,skip,require,fun,return,pays,sig,transfer},
	keywordstyle=\color{Plum}\bfseries,
	classoffset=1,
	morekeywords={},
	keywordstyle=\color{MidnightBlue}\bfseries,
}

\usepackage{amsfonts,amssymb}

\usepackage{subfigure}
\usepackage{floatpag}

\newcommand{\ifempty}[3]{%
  \ifthenelse{\isempty{#1}}{#2}{#3}%
}

\newcommand{\ifdots}[3]{%
  \ifthenelse{\equal{#1}{...}}{#2}{#3}%
}

\newcommand{\hidden}[1]{}

\newcommand{\keyterm}[1]{\textbf{\emph{#1}}}%

\makeatletter
\newcommand*{\itemequation}[3][]{%
  \item
  \begingroup
    \refstepcounter{equation}%
    \ifx\\#1\\%
    \else  
      \label{#1}%
    \fi
    \sbox0{#2}%
    \sbox2{$\displaystyle#3\m@th$}%
    \sbox4{\@eqnnum}%
    \dimen@=.5\dimexpr\linewidth-\wd2\relax
    \ifcase
        \ifdim\wd0>\dimen@
          \z@
        \else
          \ifdim\wd4>\dimen@
            \z@
          \else 
            \@ne
          \fi 
        \fi
      \@latex@warning{Equation is too large}%
    \fi
    \noindent   
    \rlap{\copy0}%
    \rlap{\hbox to \linewidth{\hfill\copy2\hfill}}%
    \hbox to \linewidth{\hfill\copy4}%
    \hspace{0pt}
  \endgroup
  \ignorespaces 
}
\makeatother

\newcommand{\Real}[1]{\mathrm{Real}}

\newcommand{\codefont}{\fontsize{9}{9}\selectfont}
\newcommand{\code}[1]{{\tt\codefont{#1}}}
\newcommand{\contract}[1]{{\tt\codefont{\txColor{#1}}}}
\newcommand{\txcode}[1]{{\text{\tt\codefont{\txColor{#1}}}}}
\newcommand{\txscriptcode}[1]{\text{\tt\codefont{\color{Plum}{#1}}}}

\newcommand{\braceleft}{\code{\{}}
\newcommand{\braceright}{\code{\}}}

\newcommand{\Eg}{E.g.\@\xspace}
\newcommand{\eg}{e.g.\@\xspace}
\newcommand{\ie}{i.e.\@\xspace}
\newcommand{\wrt}{w.r.t.\@\xspace}
\newcommand{\wwlog}{w.l.o.g.\@\xspace}

\newcommand{\emptyseq}{\varepsilon}

\def\negcaptionspace{\vspace{-10pt}}

\newenvironment{proofof}[2][]{%
  \ifempty{#1}
  {\subsection*{\normalsize Proof of~\Cref{#2}}}
  {\subsection*{\normalsize Proof of~\Cref{#2} ({#1})}}
  \label{#2-proof}
  }%
  {}

\newcommand{\procF}{\txcode{f}}

\newcommand{\op}[3]{\ensuremath{#2~\code{#1}~#3}}

\newcommand{\ifE}[3]{\code{if}~{#1}~\code{then}~{#2}~\code{else}~{#3}}

\newcommand{\nullSem}{\mathit{null}}
\newcommand{\nullE}{\code{null}}

\newcommand{\skipE}{\code{skip}}
\newcommand{\reqE}{\code{require}~}
\newcommand{\hashE}[1]{\code{H}(#1)}
\newcommand{\hashSem}[1]{\ifempty{#1}{H}{H(#1)}}

\newcommand{\secE}[1]{\code{sec}({#1})}

\newcommand{\nonceR}[1][]{r_{#1}}

\newcommand{\rk}[2]{\mathit{rev}\ifempty{#1}{}{({#1},{#2})}}

\def\pmvColor{\color{ForestGreen}}
\newcommand{\pmvFmt}[1]{{\pmvColor{\tt #1}}}

\newcommand{\pmv}[2][]{\pmvFmt{#2}_{\pmvColor{#1}}\xspace}
\newcommand{\pmva}[1][]{{\pmvColor{a_{#1}}}} 
\newcommand{\pmvA}[1][]{\pmv[{#1}]{A}} 
\newcommand{\pmvB}[1][]{\pmv[{#1}]{B}}

\newcommand{\pmvO}[1][]{\pmv[{#1}]{O}} 
\newcommand{\pmvM}{\pmv{M}} 

\newcommand{\PmvA}[1][]{\pmvFmt{\mathcal{A}}_{\pmvColor{#1}}} 
\newcommand{\PmvAfin}{\pmvFmt{\mathcal{A}}_{\pmvColor{0}}} 
\newcommand{\PmvAi}[1][]{\pmvFmt{\mathcal{A}'_{\pmvColor{\rm #1}}}} 
\newcommand{\PmvB}[1][]{\pmvFmt{\mathcal{B}}_{\pmvColor{#1}}} 
\newcommand{\PmvBi}[1][]{\pmvFmt{\mathcal{B}'_{\pmvColor{#1}}}} 
\newcommand{\PmvC}[1][]{\pmvFmt{\mathcal{C}}_{\pmvColor{#1}}} 
\newcommand{\PmvU}[1][]{\pmvFmt{\mathbb{A}}_{#1}} 

\newcommand{\txscript}{\textsc{TxScript}\xspace}

\def\txColor{\color{MidnightBlue}}

\newcommand{\txFmt}[1]{{\txColor{\sf #1}}}

\newcommand{\tx}[2][]{\txFmt{#2}_{\txColor{#1}}}
\newcommand{\txT}[1][]{\tx[#1]{X}} 
\newcommand{\txY}[1][]{\tx[#1]{Y}} 
\newcommand{\txTi}[1][]{\txFmt{X'_{\txColor{{\it #1}}}}}

\newcommand{\TxT}[1][]{\tx[#1]{\mathcal{X}}} 
\newcommand{\TxTfin}{\tx[0]{\mathcal{X}}} 
\newcommand{\TxTi}[1][]{\tx[#1]{\mathcal{X}'}} 
\newcommand{\TxY}[1][]{\tx[#1]{\mathcal{Y}}} 
\newcommand{\TxZ}[1][]{\tx[#1]{\mathcal{Z}}} 
\newcommand{\TxTS}[1][]{\vec{\tx[#1]{\mathcal{X}}}} 
\newcommand{\TxTiS}[1][]{\vec{\tx[#1]{\mathcal{X'}}}} 
\newcommand{\TxYS}[1][]{\vec{\tx[#1]{\mathcal{Y}}}} 
\newcommand{\TxU}[1][]{\tx[#1]{\mathbb{X}}} 

\newcommand{\balance}[1]{\texttt{balance}\ifempty{#1}{}{({#1})}}
\newcommand{\txsig}[1]{{#1}\;\texttt{sig}}
\newcommand{\txin}[3]{{#1}\;\texttt{pays}\;{#2\!:\!#3}}
\newcommand{\txout}[3]{\texttt{transfer}({#1},{#2\!:\!#3})}

\DeclareMathAlphabet{\mathbfsf}{\encodingdefault}{\sfdefault}{bx}{n}

\newcommand{\irule}[2]{\dfrac{#1}{#2}}

\newcommand\altxrightarrow[2][0pt]{\mathrel{\ensurestackMath{\stackengine%
  {\dimexpr#1-4.5pt}{\xrightarrow{\phantom{#2}}}{\scriptstyle\!#2\,}%
  {O}{c}{F}{F}{S}}}}

\newcommand{\sem}[2][]{\mbox{\ensuremath{\llbracket{#2}\rrbracket_{#1}}}}
\newcommand{\semcmd}[3][]{\mbox{\ensuremath{\langle{#2,#3}\rangle_{#1}}}}
\newcommand{\semexp}[2][]{\sem[#1]{#2}}

\newcommand{\Nat}{\mathbb{N}}

\newcommand{\powset}[1]{2^{#1}}
\newcommand{\bind}[2]{\nicefrac{#2}{#1}}
\newcommand{\setenum}[1]{\{#1\}}
\newcommand{\setcomp}[2]{\left\{{#1} \,\middle|\, {#2}\right\}}

\newcommand{\subseteqfin}{\subseteq_{\it fin}}

\newcommand{\idxfun}[1]{\mathbf{1}_{#1}}

\newcommand{\wmvA}[1][]{w_{#1}}

\newcommand{\WmvA}[1][]{W_{#1}}
\newcommand{\WmvAi}[1][]{W'_{#1}}
\newcommand{\WmvAii}[1][]{W''_{#1}}

\newcommand{\wal}[2]{{\mathit{\omega}_{#1}\ifempty{#2}{}{({#2})}}}

\newcommand{\WalU}[1][]{\mathbb{W}_{#1}} 
\newcommand{\finwalU}{\mbox{$\Nat^{(\TokU)}$}} 

\newcommand{\waltok}[2]{#1\!:\!#2}
\newcommand{\walenum}[1]{[#1]}
\newcommand{\walpmv}[2]{{#1}\walenum{#2}}
\newcommand{\walu}[3]{\walpmv{#1}{\waltok{#2}{#3}}}

\newcommand{\waldistrarrow}[1]{\approx_{\$}}
\newcommand{\waldistr}[4]{{#2} {\;\ifempty{#1#3}{}{{#1}} \approx_{\$} {#3}} {#4}\;}

\newcommand{\wealth}[2]{\$\ifempty{#2}{}{\mathit{\omega}_{#1}({#2})}}

\newcommand{\gain}[3]{\mathit{\gamma}_{#1}\ifempty{#2}{}{({#2},{#3})}}

\newcommand{\mall}[2]{{\kappa_{#1}\ifempty{#2}{}{({#2})}}}

\newcommand{\mev}[3]{{\mathrm{MEV\!}_{#1}\ifempty{#1#2#3}{}{\ifempty{#2#3}{}{({#2\ifempty{#3}{}{,#3}})}}}}

\newcommand{\badmev}[3]{{\mathrm{MEV^{\rm bad}\!}_{#1}\ifempty{#1#2#3}{}{\ifempty{#2#3}{}{({#2\ifempty{#3}{}{,#3}})}}}}

\newcommand{\auth}[1]{{\mathit{\alpha}\ifempty{#1}{}{({#1})}}}

\newcommand{\SysU}[1][]{\mathbb{S}_{#1}} 
\newcommand{\ContrS}[1][]{C_{#1}} 
\newcommand{\ContrU}[1][]{\mathbb{C}_{#1}} 

\newcommand{\qedex}{\ensuremath{\diamond}}

\definecolor{LightGrey}{rgb}{0.95,0.95,0.95}
\definecolor{keyword}{HTML}{7F0055}

\def\tokColor{\color{Magenta}}
\newcommand{\tokFmt}[1]{{\tokColor{\tt #1}}}
\newcommand{\tok}[2][]{\tokFmt{#2}_{\tokColor{#1}}\xspace}

\newcommand{\tokt}[1][]{{\tokColor{t_{#1}}}}    
\newcommand{\tokT}[1][]{\tok[{#1}]{T}}    
\newcommand{\tokTi}[1][]{\tok[{#1}]{T'}}
\newcommand{\TokU}{\tokFmt{\mathbb{T}}} 

\newlength\replength
\newcommand\repfrac{.1}

\newcommand\rulewidth{.6pt}
\newcommand\tdashfill[1][\repfrac]{\cleaders\hbox to \replength{%
  \smash{\rule[\arraystretch\ht\strutbox]{\repfrac\replength}{\rulewidth}}}\hfill}

\newcommand\tdotfill[1][\repfrac]{\cleaders\hbox to \replength{%
  \smash{\raisebox{\arraystretch\dimexpr\ht\strutbox-.1ex\relax}{.}}}\hfill}

\newcommand{\VarU}{{\contrColor{\mathbf{Var}}}} 
\newcommand{\VarUP}{{\contrColor{\mathbf{PVar}}}} 
\newcommand{\VarUS}{{\contrColor{\mathbf{SVar}}}} 
\newcommand{\ValU}{{\contrColor{\mathbf{Val}}}} 
\newcommand{\ValUB}{{\contrColor{\mathbf{BVal}}}} 

\newcommand{\true}{\mathit{true}}
\newcommand{\false}{\mathit{false}}

\def\contrColor{\color{MidnightBlue}}
\newcommand{\contrFmt}[1]{{\contrColor{\code{#1}}}}
\newcommand{\contrC}[1][]{\mathord{\contrFmt{C}_{\contrColor{#1}}}}
\newcommand{\contrCi}[1][]{\mathord{\contrC[#1]\contrColor{'}}}

\newcommand{\contrAdvC}[2]{\mathcal{C}} 

\def\sysColor{\color{Black}}

\newcommand{\sysFmt}[1]{{\sysColor{#1}}}
\newcommand{\sysS}[1][]{\mathord{\sysFmt{S}_{\sysColor{#1}}}}
\newcommand{\sysR}[1][]{\mathord{\sysFmt{R}_{\sysColor{#1}}}}
\newcommand{\sysSi}[1][]{\mathord{\sysColor{\sysS'_{#1}}}}
\newcommand{\sysRi}[1][]{\mathord{\sysColor{\sysR'_{#1}}}}
\newcommand{\sysSii}[1][]{\mathord{\sysColor{\sysS''_{#1}}}}

\newtheorem{thm}{Theorem}
\newtheorem{lem}{Lemma}
\newtheorem{prop}{Proposition}
\newtheorem{cor}{Corollary}
\newtheorem{defn}{Definition}
\newcommand{\qedhere}{\qed}

\newtoggle{anonymous}
\togglefalse{anonymous}

\newtoggle{arxiv}
\toggletrue{arxiv}

\makeatletter
\renewcommand\paragraph{\@startsection{paragraph}{4}{\z@}%
  {2.25ex \@plus 1ex \@minus .2ex}%
  {-0.75em}%
  {\normalfont\normalsize\bfseries}}
\makeatother

\begin{document}

\title{A theoretical basis for MEV}

\iftoggle{anonymous}{}{
\author{Massimo Bartoletti\inst{1},
Roberto Zunino\inst{2}}
\institute{
Universit\`a degli Studi di Cagliari, Cagliari, Italy
\and
Università di Trento, Trento, Italy
}
}


\maketitle

\begin{abstract}
  Maximal Extractable Value (MEV) refers to a wide class of economic attacks
  to public blockchains,
  where adversaries with the power to reorder, drop or insert transactions in a block
  can ``extract'' value from smart contracts.
  Empirical research has shown that mainstream DeFi protocols
  are massively targeted by these attacks,
  with detrimental effects on their users and on the blockchain network.
  Despite the increasing real-world impact of these attacks, their theoretical foundations remain insufficiently established.
  We propose a formal theory of MEV,
  based on a general, abstract model of blockchains and smart contracts.
  Our theory is the basis for proofs of security
  against MEV attacks.
\end{abstract}
\section{Introduction}

%
Most blockchain protocols delegate the construction of blocks to
consensus nodes that can freely pick users' transactions
from the \emph{mempool}, possibly add their own,
and propose blocks containing these transactions in a chosen order.
This arbitrariness in block construction can turn consensus nodes into adversaries, which exploit their transaction-ordering powers to maximize their gain at the expense of users.
In the crypto jargon these attacks are referred to as ``extracting'' value,
and the adversaries' gain is called \emph{Maximal Extractable Value}, or MEV.

This issue is not purely theoretical:
indeed, mainstream DeFi protocols like Automated Market Makers
and Lending Pools are common targets of MEV attacks,
which overall have led to attacks worth more than 1.2 billion dollars
so far~\cite{mevexplore}.
Notably, the profits derived from MEV attacks largely exceed
those given by block rewards and transaction fees~\cite{Daian20flash}.

MEV attacks are so profitable that currently 
most Ethereum blocks proposals
are due to centralized private relay networks that outsource
the identification of MEV opportunities to anyone,
and use their large networks of validators to include the MEV-extracting transactions in blocks~\cite{Wahrstatter23arxiv,Oz24aft}.
While this systematic MEV extraction has some benefits
(\eg, it has decreased transaction fees for users
at the expense of MEV seekers~\cite{Weintraub22flashbot}),
it is detrimental to blockchain decentralization,
transparency, and network congestion~\cite{Qin21quantifying}.

Given the practical relevance of MEV, various research efforts have focused on improving its understanding.
Most approaches are preeminently empirical,
and focus on heuristics to extract value
from certain types of contracts~\cite{Eskandari19sok,Zhou21high,Kulkarni22arxiv,BCL22fc,BabelJ0KKJ23ccs},
on the quantification of their impact in the wild
\cite{Qin21fc,Torres21frontrunner,Werner22aft,Zhou21discovery,Torres24ccs},
or on techniques to mitigate MEV attacks
\cite{breidenbach2017hydra,Baum21iacr,Heimbach22aft,baump2dex,Ciampi22cscml,Canidio24mansci,Babel24arxiv}.
All these works, however, do not answer one fundamental question:
\emph{what is MEV, exactly?}
This contrasts with fundamental principles of modern cryptography,
where formal definitions of security properties and of
adversaries' powers are essential to the study of cryptographic schemes.
In the absence of a rigorous definition, it is impossible to prove
that a contract is \emph{MEV-free}, \ie, secure \wrt MEV attacks.
Formalizing MEV is challenging, 
as it requires a complex characterization of the adversary:
\begin{inlinelist}
\item as an entity who can control the construction of blocks, 
  where they can craft and insert their own transactions;
\item whose actual identity and current wealth are immaterial \wrt MEV extraction.
\end{inlinelist}
This complexity requires a comprehensive formalization
of the adversary powers, knowledge, and their MEV attacks.
Existing MEV definitions~\cite{Babel23clockwork,Salles21formalization,Mazorra22price}
are partial (\eg, they do not formalize adversarial knowledge),
and they wrongly classify
some types of contracts (see~\Cref{sec:related}).


\paragraph*{Contributions}

We summarise our contributions as follows:
\begin{itemize}
  
\item An abstract model of contracts, 
  equipped with key economic notions like wealth and gain.
  We keep our model general in order to make it applicable to different  blockchains and contract languages.
  
\item An adversary model that includes a formalization of the adversaries' knowledge,
  \ie, the transactions that they
  can deduce by combining their private knowledge 
  with that of the mempool.
  This improves over~\cite{Babel23clockwork,Salles21formalization,Mazorra22price},
  where adversaries can craft blocks either
  by using their own knowledge or
  playing verbatim the transactions in the mempool,
  but cannot combine these two sources of information,
  so losing some potential attacks.

\item A formal definition of the MEV extractable by a \emph{given} set of users.
  We give this definition in two variants, depending on whether extracting MEV requires or not to exploit transactions in the mempool. 
  This allows us to distinguish between the ``legit'' MEV that is intended in normal contract interactions, and the ``bad'' MEV that is not.
  To the best of our knowledge, this is the first work that attempts to separate ``legit'' MEV from ``bad'' MEV.
  We show that our definitions capture new attacks, not covered by~\cite{Babel23clockwork,Salles21formalization,Mazorra22price}.

\item A formal definition of \emph{universal} MEV,
  \ie, the maximal gain that can be achieved
  by \emph{any} adversary, regardless of their actual identity
  and current wealth.
  Our MEV definition has a game-theoretic flavour:
  honest players try to minimize the damage,
  while adversaries try to maximize their gain.
  We support our notion through a theoretical study of its main
  properties, like monotonicity and finiteness;

\item Proofs for MEV-freedom:
  we assess our MEV theory
  on real-world contracts such as 
  crowdfunding, bounties, AMMs, and Lending Pools.

  
\end{itemize}

Overall, our formalization is a necessary first step towards the
construction of analysis tools for the MEV-freedom of contracts. 
\iftoggle{arxiv}
{We provide our benchmark of use cases, additional results and the proofs of our statements in the appendix.}
{Because of space constraints, we provide our benchmark of use cases, additional results and the proofs of all our statements in 
a separated technical report~\cite{BZ25arxiv}.}
\section{Blockchain model}
\label{sec:contracts}

We introduce below a formal model for reasoning about MEV. Aiming at generality and agnosticism of actual blockchains, rather than providing a concrete contract language we abstractly model contracts as state transition systems.


We assume a countably infinite set $\PmvU$ of \emph{actors}
($\pmvA, \pmvB, \ldots$),
a set $\TxU$ of \emph{transactions} ($\txT, \txTi, \ldots$),
and
a set $\TokU$ of \emph{token types} ($\tokT, \tokTi, \ldots$).
Actors are active entities, such as honest contract users or adversaries trying to extract MEV.
We assume that tokens are \emph{fungible},
\ie units of the same type are interchangeable;
NFTs are a special case of fungible tokens
that are minted in a single indivisible unit.

We use calligraphic uppercase letters for sets
(\eg, a set of actors $\PmvA$, a set of transactions $\TxT$),
and bold calligraphic uppercase for finite sequences
(\eg, a sequence of transactions $\TxTS$).
Pointwise sum of functions is denoted by $+$ and~$\sum$.

We model the token holdings of a set of actors
as a \emph{wallet} \mbox{$\wmvA : \TokU \rightarrow \Nat$},
\ie, a map from token types to non-negative integers.
A \emph{wallet state}
\mbox{$\WmvA : \PmvU \rightarrow (\TokU \rightarrow \Nat)$}
maps each actor to a wallet.
We write $\WmvA(\PmvA)$ for $\sum_{\pmvA \in \PmvA} \WmvA(\pmvA)$,
\ie, the pointwise addition of the wallets of the actors in $\PmvA$.
A \emph{blockchain state} consists of a wallet state and a contract state.
For the general results in this paper, the actual structure of contract states is immaterial, and therefore we do not specify it in~\Cref{def:contract}.
When reasoning about specific contracts, we will make this structure explicit, including data and tokens.  
We model contracts as transition systems between blockchain states, with transitions triggered by transactions.

\begin{defn}[Contract]
  \label{def:contract}
  A contract is a triple made of:
  \begin{itemize}

  \item $\SysU = \ContrU \times \WalU$, a set of \emph{blockchain states},
    where $\ContrU$ is a set of \emph{contract states}, and $\WalU$ is a set of wallet states;
  \item  $\xmapsto{\phantom{\txT}} \, : (\SysU \times \TxU) \to \SysU$,
    a partial transition function that, given a blockchain state and a transaction, gives the next blockchain state;
    
  \item $\SysU[0] \subseteq \SysU$, a set of initial blockchain states.
 
  \end{itemize}

  We denote by $\wal{}{\sysS}$ the wallet state of a blockchain state $\sysS$,
  and with $\wal{\pmvA}{\sysS}$ the wallet of an actor $\pmvA$ in $\sysS$,
  \ie, $\wal{\pmvA}{\sysS} = \wal{}{\sysS}(\pmvA)$.
\end{defn}

A transaction $\txT$ is \emph{valid} in a state $\sysS$ if $\sysS$
has an outgoing transition $\xmapsto{}$ labelled $\txT$.
In real-world contract platforms like Ethereum, invalid transactions can be included in blocks, but have no effect on the state of the contract.
To model this behaviour, we transform $\xmapsto{\;\;}$ 
into a (deterministic and) total relation
\mbox{$\xrightarrow{} \; : (\SysU \times \TxU^{*}) \rightarrow \SysU$}
between blockchain states;
$\xrightarrow{}$ is labelled with \emph{sequences} of transactions.
For an empty sequence $\emptyseq$, we 
let $\sysS \xrightarrow{\emptyseq} \sysS$,
\ie, doing nothing has no effect.
For a non-empty sequence $\TxYS \txT$, we let
$\sysS \altxrightarrow[-4pt]{\TxYS \txT} \sysSi$ when either:
\[
\sysS \altxrightarrow[-3pt]{\TxYS} \sysSii 
\text{ for some $\sysSii$, and $\sysSii \xmapsto{\txT} \sysSi$}
\quad \text{or} \quad
\sysS \altxrightarrow[-3pt]{\TxYS} \sysSi 
\text{ and $\txT$ is not valid in $\sysSi$}
\]

A state $\sysS$ is \emph{reachable} if 
$\sysS[0] \altxrightarrow[-4pt]{\TxTS} \sysS$ for some initial state $\sysS[0]$ and sequence $\TxTS$.
We implicitly assume that all the states mentioned in our results are reachable.

Our model does not allow actors to freely exchange tokens, but this is not a limitation: 
if desired, this can be encoded in the contract transition function.
This is coherent with how Ethereum handles \eg, ERC20 tokens,
whose transfer capabilities are programmed in the contract.


Our model is quite general, and also includes some behaviors that are not meaningful in practice: \eg, there may be states where the
total amount of tokens in actors' wallets is \emph{infinite}.
To rule out these cases, we require the following \emph{finite tokens axiom}, ensuring that the overall amount of tokens in wallets is finite:%
\footnote{The finite tokens axiom applies only to the wallet
state, while the tokens stored within the contract are unconstrained.
Our theory works fine under this milder hypothesis, 
since to reason about MEV we do not need to count the tokens
within the contract, but only the gain of actors.}
\begin{equation}
  \label{eq:finite-tokens}
  \textstyle
  \sum_{\tokT} \WmvA(\PmvU) (\tokT) \in \Nat
\end{equation}

As a consequence of the axiom, $\wal{\PmvA}{\sysS}$ has finite support for all $\PmvA$ and $\sysS$.
Hereafter, we denote by $\finwalU$ the set of finite-support functions from $\TokU$ to $\Nat$.


Measuring the effect of an attack to a contract
requires to estimate the \emph{wealth} of the adversary
before and after the attack.
To account for the fact that different token types can have different prices, we assume an additive function $\wealth{}{}$ that, given a wallet $\wmvA$, determines its \emph{wealth} $\wealth{}{}\wmvA$.%
\footnote{Note that $\pmvA$'s wealth only depends on $\pmvA$'s wallet, neglecting other parts of the state.
  As a result, actors with the same tokens have the same wealth,
  and wealth is insensitive to price fluctuations.  
  This is because our notion of MEV is designed
  to capture attacks occurring in a \emph{single} block.
  We discuss long-range attacks in~\Cref{sec:limitations}.}

\begin{defn}[Wealth]
  \label{def:wealth}
  We say that $\wealth{}{} : \finwalU \rightarrow \Nat$
  is a \emph{wealth function} if
  $\wealth{}{}(\wmvA[0] + \wmvA[1]) = \wealth{}{}\wmvA[0] + \wealth{}{}\wmvA[1]$  
  holds for all $\wmvA[0], \wmvA[1] \in \finwalU$.
\end{defn}


Let $\idxfun{\tokT}$ be the wallet that contains exactly one token of type $\tokT$.
Then, we can write any wallet $\wmvA$ as the (potentially infinite) sum:
\begin{equation}
  \label{eq:wealth:token-prices:wal}
  \textstyle
  \wmvA \; = \; \sum_{\tokT} \wmvA(\tokT) \cdot \idxfun{\tokT}
\end{equation}

If $\wmvA$ has finite support,
then the sum in~\eqref{eq:wealth:token-prices:wal}
has only a finite number of non-zero terms.
From additivity of $\wealth{}{}$, it follows that the wealth of a (finite-support) $\wmvA$ is the sum of the amount of each token
$\tokT$ in $\wmvA$, times the \emph{price} of $\tokT$, \ie, $\wealth{}{}\idxfun{\tokT}$:
\begin{equation}
  \label{eq:wealth:token-prices:wealth}
  \wealth{}{} \wmvA
  \; = \;
  \textstyle \sum_{\tokT} \wealth{}{} \big( \wmvA(\tokT) \cdot \idxfun{\tokT} \big)
  \; = \;
  \textstyle \sum_{\tokT} \wmvA(\tokT) \cdot \wealth{}{}\idxfun{\tokT}
\end{equation}



We measure the success of a MEV attack in terms of \emph{gain},
\ie, the difference of the attackers' wealth before and after the attack.

\begin{defn}[Gain]
  \label{def:gain}
  We define the gain of $\PmvA$ upon performing
  a sequence $\TxTS$ of transactions from state $\sysS$ as
  \(
  \gain{\PmvA}{\sysS}{\TxTS}
  =
  \wealth{\PmvA}{\sysSi} - \wealth{\PmvA}{\sysS}
  \)
  if
  $\sysS \xrightarrow{\TxTS} \sysSi$.
\end{defn}

The following proposition establishes some basic properties of gain.
In particular, the maximal gain can always be extracted by a \emph{finite} set of actors. 

\begin{prop}
  \label{prop:gain}
  $\gain{\PmvA}{\sysS}{\TxTS}$ is always defined and has a finite integer value, given by
  \(
  \gain{\PmvA}{\sysS}{\TxTS} 
  = 
  \textstyle \sum_{\pmvA \in \PmvA} \gain{\pmvA}{\sysS}{\TxTS}
  \).
  Furthermore, there exists 
  $\PmvAfin \subseteqfin \PmvA$ (finite subset of $\PmvA$) such that,
  for all $\PmvB$, if $\PmvA[0] \subseteq \PmvB \subseteq \PmvA$ then 
  $\gain{\PmvAfin}{\sysS}{\TxTS} = \gain{\PmvB}{\sysS}{\TxTS}$.
 \end{prop}

\section{Maximal Extractable Value}
\label{sec:mev}

\iftoggle{arxiv}{%
In this~\namecref{sec:mev} we introduce our definition of MEV.%
\footnote{In the original terminology introduced by~\cite{Daian20flash} the ``M'' of MEV stands for ``miner''. 
However, following Ethereum's transition from proof-of-work to proof-of-stake in September 2022 --- and the corresponding renaming of miners into validators --- most sources now interpret ``M'' as ``maximal''.}
}{}
Formalizing MEV is challenging, because
it requires a twofold characterization of the adversary as
a set of actors with the ability to reorder,
drop and insert transactions, and
whose actual identity and wealth are immaterial to MEV extraction.
We then find it convenient to divide our formalization into three steps:
\begin{enumerate}

\item We define the set of transactions $\mall{\PmvA}{\TxT}$
  that actors $\PmvA$ can \emph{deduce}
  by combining their private knowledge with that of the mempool~$\TxT$ (Def.~\ref{def:mall}).
  
\item We define the MEV 
  of a \emph{given} set of actors $\PmvA$
  in a state $\sysS$ and mempool $\TxT$
  as the \emph{maximal gain} that $\PmvA$ can achieve from $\sysS$
  by firing a sequence of transactions in their deducible knowledge
  $\mall{\PmvA}{\TxT}$ (Def.~\ref{def:mev}).
  We also provide a variant, dubbed ``bad'' MEV, which models the case where the attack is only possible by exploiting transactions in the mempool.
  
\item We define the \emph{universal} MEV
  in a state $\sysS$ and mempool $\TxT$ as the MEV that an \emph{arbitrary} set of actors
  can achieve regardless of identity,
  only assuming they have (or can buy) the tokens needed to carry out the attack (Def.~\ref{def:adv-mev}).


\end{enumerate}


\subsection{Adversary model}
\label{sec:tx-inference}




Given a mempool $\TxT$, adversaries $\PmvA$ can craft new transactions by mauling the transactions data in~$\TxT$ 
(\eg, method arguments in Ethereum transactions).
To this goal, $\PmvA$ can reuse any piece of data in $\TxT$,
with the constraints that they cannot forge signatures,
and cannot deduce any value that is not efficiently
computable from the previous ones (\eg, inverting a hash).

In our abstract model, we generalize this inference with an axiomatization of the set $\mall{\PmvA}{\TxT}$ of transactions deducible by the adversary $\PmvA$ from a given mempool $\TxT$.
We start by requiring \keyterm{extensivity}, \keyterm{idempotence} and \keyterm{monotonicity} on $\TxT$,  so making $\mall{\PmvA}{\cdot}$ an upper closure operator for any $\PmvA$.
In particular, these axiom imply that $\mall{\PmvA}{\TxT}$ include all the transactions in the mempool $\TxT$, and that larger mempools lead to larger adversarial inferences.
The \keyterm{continuity} axiom is a standard structural requirement: in our theory, it is pivotal to prove that MEV can always be extracted from a finite mempool.
The \keyterm{finite causes} axiom ensures that any \emph{finite} set of transactions can be deduced by a \emph{finite} set of actors that only use their private knowledge. Abstractly, the private knowledge is the set of transactions deducible from an empty mempool $\TxT$ (in practice, this corresponds to the set of transactions that $\PmvA$ can craft by using their private keys).
The \keyterm{private knowledge} axiom states that
a larger private knowledge requires a larger set of actors:
hence, if two sets of actors can deduce exactly the same transactions, then they must be equal.
Finally, the \keyterm{no shared secrets} axiom formalises a separation between the private knowledge of different actors: namely, it implies that if a transaction can be deduced by two disjoint sets of actors, then it can be deduced by anyone.

\begin{defn}[Transaction deducibility]
  \label{def:mall}
  We say that 
  $\mall{}{} : \powset{\PmvU} \times \powset{\TxU} \rightarrow \powset{\TxU}$
  is a transaction deducibility function if it satisfies the following axioms:
  \begin{description}

  \item[Extensivity] 
    $\TxT \subseteq \mall{\PmvA}{\TxT}$

  \item[Idempotence] 
    $\mall{\PmvA}{\mall{\PmvA}{\TxT}} = \mall{\PmvA}{\TxT}$

  \item [Monotonicity] 
    if $\PmvA \subseteq \PmvAi$, $\TxT \subseteq \TxTi$, 
    then $\mall{\PmvA}{\TxT} \subseteq \mall{\PmvAi}{\TxTi}$

  \item [Continuity]
  for all chains $\TxT[0] \subseteq \TxT[1] \subseteq \TxT[2] \subseteq \cdots $,    
    $\mall{\PmvA}{\bigcup_{i\in\Nat} \!\TxT[i]} = \bigcup_{i\in\Nat} \!\mall{\PmvA}{\TxT[i]}$
  
    
  \item [Finite causes] 
    $\forall \TxTfin \subseteqfin \TxU . \ \exists \PmvAfin \subseteqfin \PmvU . \ \TxTfin \subseteq \mall{\PmvAfin}{\emptyset}$

  \item [Private knowledge] if $\mall{\PmvA}{\emptyset} \subseteq \mall{\PmvAi}{\emptyset}$, then $\PmvA \subseteq \PmvAi$
    
  \item [No shared secrets] 
    $\mall{\PmvA}{\TxT} \cap \mall{\PmvB}{\TxT} \subseteq \mall{\PmvA \cap \PmvB}{\TxT}$

  \end{description}
\end{defn}

We illustrate $\mall{\PmvA}{\TxT}$ through an example.
For simplicity, in the contract language used in our examples we drop all the features of Solidity that are inessential to the understanding of MEV;
still, the language is expressive enough to express real-world use cases like those found in DeFi 
\iftoggle{arxiv}{(see~\Cref{sec:examples}).}{see~\cite{BZ25arxiv}.}
Our contract language has a formal semantics.
While this semantics is crucial to prove the presence or absence of MEV in our benchmark of use cases, for now it will be sufficient to rely on its intuitive understanding.

\begin{figure}[t]
  \begin{lstlisting}[language=txscript,morekeywords={BadHTLC,commit,reveal,timeout},classoffset=3,morekeywords={a,b,A,Oracle,verifier},keywordstyle=\pmvColor,classoffset=4,morekeywords={t,T},keywordstyle=\tokColor,frame=single]
contract BadHTLC {
  commit(a pays 1:T,b,c) { // a must send 1:T to the contract
    require balance(T)==1; commitment=c; // prevents multiple commits
  }
  reveal(a sig,y) { // a must sign the transaction and reveal the secret y
    require balance(T)>0 && H(y)==commitment; // y must be a preimage
    transfer(a,balance(T):T); // send all T balance to a
  }
  timeout(a sig,Oracle sig) { // a and Oracle must sign the transaction
    require balance(T)>0; // the contract must have some tokens T
    transfer(a,balance(T):T); // send all T balance to a
  }
}
  \end{lstlisting}
  \negcaptionspace  
  \caption{A Hash Time Locked Contract.}
  \label{fig:badhtlc}
\end{figure}

\begin{example}
  \label{ex:badhtlc}  
  \label{ex:txscript:badhtlc:mall}
  The \txcode{BadHTLC} contract in~\Cref{fig:badhtlc}
  implements a Hash-Time Locked Contract, where
  a committer promises that she will either reveal a secret
  within a certain deadline, or pay a penalty of \mbox{$\waltok{1}{\tokT}$} to anyone after the deadline.
  Let $\sysS$ be a state where the secret has been committed to $H(s)$ but not revealed yet,
  and assume the mempool $\TxT$ contains a transaction
  \mbox{$\txT[\pmvA] = \txcode{reveal}(\txsig{\pmvA},s)$}
  sent by $\pmvA$ to redeem the deposit.
  Since the secret $s$ is public in the mempool,
  any adversary $\pmvM$ can craft a transaction
  $\txT[\pmvM] = \txcode{reveal}(\txsig{\pmvM},s)$
  by combining their own knowledge (to provide $\pmvM$'s signature)
  with that of $\TxT$ (to provide~$s$).
  Hence, $\txT[\pmvM] \in \mall{\setenum{\pmvM}}{\TxT}$, and
  so $\pmvM$ can extract MEV
  by front-running $\txT[\pmvA]$ with $\txT[\pmvM]$.
  Note instead that $\pmvM$ \emph{alone} cannot deduce the transaction that would allow her to trigger the timeout
  \iftoggle{arxiv}
  { (see~\Cref{ex:htlc} for details).}
  { (see~\cite{BZ25arxiv} for details).}
  We remark that this attack relies on combining private and mempool knowledge, which does not seem to be properly accounted for in
  current MEV formalizations~\cite{Babel23clockwork,Salles21formalization,Mazorra22price}.
  \hfill\qedex  
\end{example}

\Cref{lem:mall:basic} establishes some key properties of $\mall{}{}$, which will be instrumental to prove more complex properties about MEV.
\Cref{lem:mall:inj} states that different sets of actors have a different private knowledge.
\Cref{lem:mall:cap} implies that transactions that can be deduced by disjoint sets of actors can be deduced by anyone. 
\Cref{lem:mall:cup} states that two groups of actors joining forces could infer more transactions than they could independently infer.
\Cref{lem:mall:subseteq-mall-mall} states that a group $\PmvB$ that can exploit both a mempool $\TxT$ and the inference of a larger group $\PmvA$ on a smaller mempool $\TxY$ cannot infer more transactions than $\PmvA$  infer from $\TxT$.
Remarkably, \Cref{lem:mall:finite-tx} rules out the case where, to deduce a transaction $\txT$,
a set of actors needs to combine knowledge from an \emph{infinite} mempool $\TxT$ (the proof exploits the continuity of $\mall{}{}$).

\begin{prop}
  \label{lem:mall:basic}
  For all $\PmvA$, $\PmvB$, $\TxT$ and $\TxY$, we have that:
  \begin{enumerate}[\rm(1)]

  \item \label{lem:mall:inj}
    if $\PmvA \neq \PmvB$, then $\mall{\PmvA}{\emptyset} \neq \mall{\PmvB}{\emptyset}$


  \item \label{lem:mall:cap}
  $\mall{\PmvA}{\TxT} \cap \mall{\PmvB}{\TxT} = \mall{\PmvA \cap \PmvB}{\TxT}$

  \item \label{lem:mall:cup} \label{lem:mall:mall}
    $\mall{\PmvA}{\TxT} \cup \mall{\PmvB}{\TxT} \subseteq \mall{\PmvA}{\mall{\PmvB}{\TxT}} \subseteq \mall{\PmvA \cup \PmvB}{\TxT}$


  \item \label{lem:mall:subseteq-mall-mall}
    if $\PmvB \subseteq \PmvA$ and $\TxY \subseteq \TxT$, then
    $\mall{\PmvB}{\mall{\PmvA}{\TxY} \cup \TxT} \subseteq \mall{\PmvA}{\TxT}$


  \item \label{lem:mall:finite-tx}
    $\forall \txT \in \mall{\PmvA}{\TxT} .\ \exists \TxTfin \subseteqfin \TxT .\ \txT \in \mall{\PmvA}{\TxTfin}$
    
  \end{enumerate}
\end{prop}

\subsection{MEV extractable by a given set of actors}


The axiomatization of adversarial knowledge is the core of our MEV definition.
Namely, the MEV of a \emph{given} set of actors $\PmvA$ is the maximal gain
that $\PmvA$ can achieve by firing a sequence of transactions
\emph{deducible} by $\PmvA$ using their private knowledge and that of the mempool.

\begin{defn}[MEV]
  \label{def:mev}
  The MEV extractable \emph{by a set of actors $\PmvA$}
  from a mempool $\TxT$ in a state $\sysS$ is given by:
  \begin{align}
    \label{eq:mev}
    \mev{\PmvA}{\sysS}{\TxT}
    & =
    \max \setcomp{\gain{\PmvA}{\sysS}{\TxYS}}{\TxYS \in \mall{\PmvA}{\TxT}^*}   
  \end{align}
\end{defn}

By allowing $\PmvA$ to fire arbitrary bundles in $\mall{\PmvA}{\TxT}^*$ 
(the set of finite sequences of transactions in $\mall{\PmvA}{\TxT}$), 
we are actually empowering $\PmvA$ with the ability to reorder,
drop and insert transactions.
This is coherent with the practice,
where miners/validators are commonly in charge for assembling blocks.%
\footnote{Some blockchain networks instead do not allow the current leader node to propose a block, but use special protocols that ensure a fair ordering of transactions~\cite{Kelkar20crypto,LiPournaras24arxiv,Raikwar23brains}.}

Note that the $\max$ in~\eqref{eq:mev} may not exist:
for instance, consider a contract where each transaction
(fireable by anyone)
increases by $\waltok{1}{\tokT}$ the tokens in $\pmvA$'s wallet.
The wallet states reachable in this contract satisfy
the finite tokens axiom~\eqref{eq:finite-tokens},
but the MEV of $\pmvA$ is unbounded, because for each fixed $n$,
there exists a reachable state where $\pmvA$'s gain is greater than $n$.
A sufficient condition for the existence of the $\max$ in~\eqref{eq:mev}
is that, in any reachable state, the wealth of all actors
is bounded by a constant:
\begin{align}
  & \label{eq:bounded-wealth-axiom}
  \forall \sysS[0] \in \SysU[0] . \;
  \exists n . \;
  \forall \sysS . \;\;
  \sysS[0] \xrightarrow{} \cdots \xrightarrow{} \sysS
  \implies
  \wealth{\PmvU}{\sysS} < n
\end{align}
Hereafter, we assume that contracts satisfy~\eqref{eq:bounded-wealth-axiom},
namely they are \emph{$\wealth{}{}$-bounded}.
We now establish some key properties of MEV.
First, MEV is always defined for $\wealth{}{}$-bounded contracts.
\begin{prop}
  \label{prop:mev:defined}
  $\mev{\PmvA}{\sysS}{\TxT}$ is defined and has a non-negative value.
\end{prop}

The $\mev{}{}{}$ is preserved by removing from $\TxT$
all the transactions that the actors $\PmvA$ can generate by themselves:
\begin{prop}
  \label{prop:mev:cut}
  \(
  \mev{\PmvA}{\sysS}{\TxT} = \mev{\PmvA}{\sysS}{\TxT \setminus \mall{\PmvA}{\emptyset}}
  \)
\end{prop}


$\mev{}{}{}$ is monotonic \wrt the mempool $\TxT$.
This follows from the monotonicity of transactions deducibility $\mall{}{}$,
since a wider knowledge gives more opportunities to increase one's gain.

\begin{prop} 
  \label{prop:mev:monotonic-on-tx}
  If $\TxT \subseteq \TxTi$, then \
  \(
  \mev{\PmvA}{\sysS}{\TxT} \leq \mev{\PmvA}{\sysS}{\TxTi}
  \).
\end{prop}

Perhaps surprisingly, MEV is \emph{not} monotonic
\wrt the set $\PmvA$ of actors who are extracting it.
For instance, if $\pmvA$ has a positive gain by firing a transaction
that yields a negative opposite gain for $\pmvB$
(and $\pmvB$ has no other ways to have a positive gain),
then the MEV of $\setenum{\pmvA}$ is positive,
while that of $\setenum{\pmvA,\pmvB}$ is zero.


In general, MEV is \emph{not} even monotonic \wrt the amount of tokens in wallets,
\ie, being richer does not always increase one's ability of extracting MEV.
In particular, $\pmvA$ might be able to extract MEV in a state
\mbox{$\walpmv{\pmvA}{\wmvA} \mid \sysS$},
but not in a state \mbox{$\walpmv{\pmvA}{\wmvA+\wmvA[\Delta]} \mid \sysS$}.
For instance, this may happen when the contract enables
the MEV-extracting transaction
only in states containing an exact number of tokens in users' wallets.
In fact, in most real-world contracts the effect of transactions never depends
on tokens which are not controlled by the contract.
We formalise this property of contracts by requiring that each transaction
enabled in a certain wallet state $\WmvA$ produces the same effect
in a ``richer'' wallet state $\WmvA+\WmvA[\Delta]$.

\begin{defn}[Wallet-monotonic contract]
  A contract is \emph{wallet-monotonic} if
  for all $\WmvA,\WmvAi,\WmvA[\Delta]$, $\contrC$, $\contrCi$ and $\txT$:
  \[
  (\WmvA,\contrC) \xmapsto{\txT} (\WmvAi,\contrCi)
  \implies
  (\WmvA+\WmvA[\Delta],\contrC) \xmapsto{\txT} (\WmvAi+\WmvA[\Delta],\contrCi)  
  \]
\end{defn}

This naturally extends to sequences of valid transactions.
For this class of contracts, MEV is monotonic \wrt wallets.

\begin{prop} 
  \label{prop:mev:monotonic-on-wal}  
  Let $\sysS = (\WmvA,\contrC)$ and let
  $\sysS[\Delta] = (\WmvA+\WmvA[\Delta],\contrC)$.
  If the contract is wallet-monotonic, then \
  \(
  \mev{\PmvA}{\sysS}{\TxT} \leq \mev{\PmvA}{\sysS[\Delta]}{\TxT}
  \).
\end{prop}


In general, a single actor could not be able to extract MEV,
since the contract could require the interaction
between multiple actors in order to trigger a payment. 
\Cref{th:mev:finite-part} shows that
the (maximal!) MEV can always be obtained by a \emph{finite} set of actors.

\begin{prop} 
  \label{th:mev:finite-part}
  For all $\PmvA$, $\TxT$ and $\sysS$,
  there exists $\PmvAfin \subseteqfin \PmvA$
  such that, for all $\PmvB$, if $\PmvAfin \subseteq \PmvB \subseteq \PmvA$
  then: \
  \(
  \mev{\PmvAfin}{\sysS}{\TxT} = \mev{\PmvB}{\sysS}{\TxT}
  \).
\end{prop}

We also show that MEV can always be extracted from
a \emph{finite} mempool.
This follows from the continuity of $\mall{}{}$,
and in particular from its consequence 
\Cref{lem:mall:basic}\ref{lem:mall:finite-tx},
which ensures that each transaction in the sequence used to obtain
the $\max$ gain can be deduced from a finite subset of the mempool.

\begin{prop}[Mempool finiteness]
  \label{th:mev:finite-tx}
  For all $\PmvA$, $\TxT$ and $\sysS$, there exists some $\TxTfin \subseteqfin \TxT$ such that
  $\mev{\PmvA}{\sysS}{\TxTfin} = \mev{\PmvA}{\sysS}{\TxT}$.
\end{prop}

Not all MEV is always considered an attack:  
\eg, the MEV that derives from arbitrage on AMMs and liquidations on Lending Pools is rather considered an incentive for users to keep the contract aligned with its ideal functionality.
In these cases, the MEV is extracted without using the mempool.
To isolate the part of MEV that is agreeably considered an attack, we remove from the overall $\mev{\PmvA}{\sysS}{\TxT}$
the part that can be extracted without knowledge of the mempool, \ie, $\mev{\PmvA}{\sysS}{\emptyset}$.
We dub this new notion as ``bad MEV''.




\begin{defn}[Bad MEV]
\label{def:badmev}
  The ``bad MEV'' extractable \emph{by a set of actors $\PmvA$}
  from a mempool $\TxT$ in a state $\sysS$ is given by:
  \begin{align}
    \label{eq:badmev}
    \badmev{\PmvA}{\sysS}{\TxT}
    & =
    \mev{\PmvA}{\sysS}{\TxT}
    - 
    \mev{\PmvA}{\sysS}{\emptyset}
  \end{align}
\end{defn}

\begin{prop}
\label{prop:badlem}
All the previous results about $\mev{}{}{}$, except~\Cref{prop:mev:monotonic-on-wal}, also hold for $\badmev{}{}{}$.
Furthermore, $\badmev{\PmvA}{\sysS}{\TxT} \leq \mev{\PmvA}{\sysS}{\TxT}$.
\end{prop}

\subsection{Universal MEV}
\label{sec:adv-mev}
\label{sec:mev-freedom}

\Cref{def:mev} parameterises MEV over a set of actors $\PmvA$.
In this way, the same state $(\sysS,\TxT)$ could admit
different MEVs for different sets of actors.
This dependency on the set of actors contrasts with the practice, 
where the actual identity of miners or validators
is unrelated to their ability to extract MEV.

For instance, consider the \txcode{Whitelist} contract
in~\Cref{fig:whitelist}.
In any state $\sysS$ where $\pmvM$ has at least $\waltok{1}{\tokT}$
and the contract has $\waltok{n}{\tokT}$ with $n>0$,
any set of actors $\PmvB$ including $\pmvA$ has MEV.
More precisely, $\mev{\PmvB}{\sysS}{\TxT} = n \cdot \wealth{}{}\idxfun{\tokT}$.
However, this way of extracting MEV is \emph{not} considered
an attack in practice,
since the recipient of the tokens is not arbitrary,
but an actor ($\pmvA$) who is \emph{hard-coded} in the contract.
By contrast, the contract \txcode{Blacklist} is attackable,
provided that the adversary $\PmvB$ includes some $\pmvM \neq \pmvA$
who has at least $\waltok{1}{\tokT}$.
The fact that the hard-coded actor $\pmvA$ cannot extract MEV is irrelevant,
since the adversary can easily create a pseudonym which is different
from the blacklisted ones.


\begin{figure}[t]
  \begin{lstlisting}[language=txscript,morekeywords={Whitelist,Blacklist,pay,Bank,deposit,xfer,wdraw},classoffset=3,morekeywords={a,a0,a1,b,o,A,B},keywordstyle=\pmvColor,classoffset=4,morekeywords={t,T},keywordstyle=\tokColor,frame=single]
contract Whitelist {
  pay(a pays 1:T) { require a==A; transfer(a,balance(T):T); }
}

contract Blacklist {
  pay(a pays 1:T) { require a!=A; transfer(a,balance(T):T); }
}  

contract Bank {
  deposit(a pays amt:T) { // a sends amt:T to the contract
    if acct[a]==null then acct[a]=amt else acct[a]=acct[a]+amt
  }
  xfer(a sig,amt,b) { // a transfers amt:T to b
    require acct[a]>=amt && acct[b]!=null && amt>0;
    acct[a]=acct[a]-amt; acct[b]=acct[b]+amt
  }
  wdraw(a sig,amt) { // a withdraws amt:T
    require acct[a]>=amt && amt>0;
    acct[a]=acct[a]-amt; transfer(a,amt:T);
  }
}
  \end{lstlisting}
  \negcaptionspace
  \caption{\txcode{Whitelist}, \txcode{Blacklist} and \txcode{Bank} contracts.}
  \label{fig:whitelist}
  \label{fig:blacklist}
  \label{fig:bank}
\end{figure}


As another example of a non-attack,
consider the \txcode{Bank} contract in~\Cref{fig:bank},
which allows users to deposit and withdraw tokens,
and to transfer them to others.
Let $\sysS$ be a state where $\pmvA$ has deposited $\waltok{n}{\tokT}$
in the contract, while the balances of the other users are zero.
We have that $\mev{\setenum{\pmvA}}{\sysS}{\emptyset} = n \cdot \wealth{}{}\idxfun{\tokT}$
and in general any set of actors including $\pmvA$ can extract MEV.
However, even this way of extracting MEV is \emph{not} considered an attack,
since any adversary $\PmvB$ not including $\pmvA$ has
$\mev{\PmvB}{\sysS}{\TxT} = 0$, for every mempool $\TxT$
(unless $\TxT$ includes an explicit \txcode{xfer} from $\pmvA$ to $\PmvB$).
Unlike in the \txcode{Blacklist} example,
the \txcode{Bank} has no hard-coded names,
but some names become bound in the contract states upon transactions.
For instance, after $\pmvA$ has deposited $\waltok{n}{\tokT}$,
we have
$\walpmv{\contract{Bank}}{\waltok{n}{\tokT}, \code{acct} = \setenum{\pmvA \mapsto n}}$.

In general, if the ability to extract MEV is subject to the existence
of specific actors in the set $\PmvA$, this is \emph{not}
considered an attack.
Even when the identity of actors is immaterial,
the amount of tokens in their wallets is not:
\eg, actors may lose the capability of extracting MEV
when spoiled from their tokens.
Therefore, we consider as universal MEV the $\max$ gain
that can be extracted by the adversary
after a suitable \emph{redistribution} of tokens in wallets.
To stay on the safe side, we always consider the optimal token  
redistribution for the adversary.
In this way, we never consider a state $\sysS$ as MEV-free
just because the adversary has not enough tokens in $\sysS$ to carry the attack:
this would be unsafe, since the attack could have been possible 
in a state $\sysSi$ where the adversary has redistributed the tokens
in their wallets.
Formally, a \emph{token redistribution} is a relation
$\waldistr{}{\sysS}{}{\sysSi}$ which holds
when all the tokens in the wallets in $\sysS$ are reassigned in $\sysSi$.

\begin{defn}[Token redistribution]
  \label{def:wallet-redistribution}
  Let $\sysS = (\WmvA,\contrC)$,\ \mbox{$\sysSi = (\WmvAi,\contrCi)$}.
  We write $\waldistr{}{\sysS}{}{\sysSi}$ when
  $\WmvA(\PmvU) = \WmvAi(\PmvU)$ and $\contrC = \contrCi$.
\end{defn}

We now formalise the key notion of \emph{universal MEV}.
States with no universal MEV are called \emph{MEV-free}.

\begin{defn}[Universal MEV]
  \label{def:adv-mev}
  \label{def:mev-free}
  The \emph{universal MEV} extractable 
  from a mempool $\TxT$ in a state $\sysS$ is given by:  
  \begin{equation} \label{eq:adv-mev}
  \mev{}{\sysS}{\TxT}
  =
  \min_{\PmvB \text{ cofinite}} \hspace{5pt}
  \max_{\scriptsize \begin{array}{c} \PmvA \subseteq \PmvB \\ \waldistr{}{\sysS}{}{\sysSi} \end{array}} \mev{\PmvA}{\sysSi}{\TxT}
  \end{equation}  
  We say that
  $(\sysS,\TxT)$ is \emph{$\mev{}{}{}$-free} when $\mev{}{\sysS}{\TxT} = 0$.
  We define $\badmev{}{\sysS}{\TxT}$ and \emph{$\badmev{}{}{}$-free} similarly.
\end{defn}



To ensure that the identities of actors extracting MEV are immaterial,
in~\eqref{eq:adv-mev} we take the \emph{minimum}
\wrt all sets $\PmvB$ of actors.
We restrict to \emph{infinite} sets $\PmvB$ 
to grant the attacker an unbounded amount of fresh identities,
which can be used to avoid the ones handled in a special way by the contract.
More specifically, since real-world contracts treat,  
in each state, only a finite number of actors as special,
we let $\PmvB$ range over \emph{cofinite} sets.
Once the set $\PmvB$ of adversaries is fixed,
we take the \emph{maximum} MEV of $\PmvA \subseteq \PmvB$
\wrt all the possible token redistributions.
In this way, we ensure that the adversary has enough tokens to carry
the attack, if any.
Note that \eqref{eq:adv-mev} follows the \emph{minimax} principle
of game theoretic definitions.
Intuitively, the honest players choose $\PmvB$ in the $\min{}$
so to prevent the adversary from using privileged identities.
Then, the adversary chooses, in the $\max{}$, the identities $\PmvA$ from $\PmvB$
which are actually used for the attack:
this allows the adversary to remove the actors with negative gain.



The universal MEV is always defined
under a strengthened $\wealth{}{}$-boundedness assumption
which considers token redistributions:
\begin{align}
  & \label{eq:bounded-wealth-axiom-waldistr}
  \forall \sysS[0] \in \SysU[0] . \;
  \exists n . \;
  \forall \sysS . \;\;
  \sysS[0] \; \xrightarrow{}_{\waldistrarrow{}} \cdots \xrightarrow{}_{\waldistrarrow{}} \; \sysS
  \implies
  \wealth{\PmvU}{\sysS} < n
\end{align}
where the relation $\xrightarrow{}_{\waldistrarrow{}}$
allows to redistribute tokens at each step.
\iftoggle{arxiv}{}
{Additionally, universal MEV is monotonic with respect to mempools and wallets.}

\begin{prop}
  \label{prop:adv-mev:defined}
  For all $\TxT$ and $\sysS$ satisfying~\eqref{eq:bounded-wealth-axiom-waldistr},
  $\mev{}{\sysS}{\TxT}$ and
  $\badmev{}{\sysS}{\TxT}$
  are defined and have a non-negative value.
\end{prop}

\iftoggle{arxiv}
{The following two theorems establish that universal MEV is 
monotonic with respect to mempools and wallets.}
{}

\begin{thm} 
  \label{prop:adv-mev:monotonic-on-tx}
  If $\TxT \subseteq \TxTi$
  then
  $\mev{}{\sysS}{\TxT} \leq \mev{}{\sysS}{\TxTi}$
  (similarly for $\badmev{}{}{}$).
\end{thm}

\begin{thm} 
  \label{prop:adv-mev:monotonic-on-wal}  
  Let $\sysS = (\WmvA,\contrC)$ and
  let $\sysS + \WmvA[\Delta]$ be $(\WmvA+\WmvA[\Delta],\contrC)$.
  If the contract is wallet-monotonic, then:\
  \(
  \mev{}{\sysS}{\TxT} \leq \mev{}{\sysS + \WmvA[\Delta]}{\TxT}
  \).
\end{thm}




\subsection{Proving MEV-freedom}


\begin{table}[t]
\scriptsize
    \centering
    \caption{MEV analysis of our benchmark of contracts.}
    \label{tab:benchmark}    
    \begin{tabular}{|c|c|c||c|c|c|}
    \hline
    \textbf{Contract} 
    & $\mev{}{}{}$-free? 
    & $\badmev{}{}{}$-free?
    &
    \textbf{Contract} 
    & $\mev{}{}{}$-free?
    & $\badmev{}{}{}$-free?
    \\
    \hline
    Bad HTLC \iftoggle{arxiv}{(\ref{ex:htlc})}{} & \xmark & \xmark
    &
    Crowdfund \iftoggle{arxiv}{(\ref{ex:crowdfund})}{} & \cmark & \cmark 
    \\
    HTLC \iftoggle{arxiv}{(\ref{ex:htlc})}{} & \cmark & \cmark
    &
    AMM \iftoggle{arxiv}{(\ref{ex:amm})}{} & \xmark & \xmark 
    \\
    Whitelist \iftoggle{arxiv}{(\ref{ex:whitelist})}{} & \cmark & \cmark 
    &
    Price Bet \cite{Babel23clockwork,BMZ24fc} & \xmark & \cmark 
    \\
    Blacklist \iftoggle{arxiv}{(\ref{ex:blacklist})}{} & \xmark & \cmark 
    &
    Naïve bounty \iftoggle{arxiv}{(\ref{ex:bounty})}{} & \xmark & \xmark
    \\
    Bank \iftoggle{arxiv}{(\ref{ex:bank})}{} & \cmark & \cmark
    &
    Bounty \iftoggle{arxiv}{(\ref{ex:bounty})}{} & \cmark & \cmark
    \\
    Coin Pusher \iftoggle{arxiv}{(\ref{ex:coinpusher})}{} & \xmark & \xmark
    &
    Lending Pool \iftoggle{arxiv}{(\ref{ex:lp})}{} & \xmark & \cmark
    \\
    \hline
    \end{tabular}
\end{table}

\iftoggle{arxiv}
{We now apply our theory to study MEV-freedom of contracts.}
{We now apply our MEV theory.}
\iftoggle{arxiv}
{Here we just summarize the results of our analysis in~\Cref{tab:benchmark}, and we refer to~\Cref{sec:examples} for the technical details.}
{Because of space constraints, here we just summarize the results of our analysis in~\Cref{tab:benchmark} and refer to~\cite{BZ25arxiv} for the technical details.}
In the column ``$\mev{}{}{}$-free?'', we mark with a \xmark\ those cases where a (universal) MEV attack exists with respect to some state and mempool, and with \cmark\ when no such attack exist. 
The column ``$\badmev{}{}{}$-free?'' uses a similar notation.
 
The \xmark\ in the AMM is witnessed by a sandwich attack in a state where the AMM is in balanced state (\ie, the internal exchange rate equals to that of an external page oracle).
This is a case of ``bad'' MEV, since the attack exploits the mempool.
When the AMM is unbalanced, there is a case of ``legit'' MEV, \ie, anyone can perform arbitrage and have a positive MEV with an empty mempool. 
The \xmark\ in the Lending Pool is witnessed by liquidations of under-collateralized borrowers. 
%
The Price Bet contract in~\cite{Babel23clockwork,BMZ24fc} shows an interesting case where (universal) MEV is extractable without exploiting the mempool.
In this contract, a player bets on the future exchange rate between two tokens, determined through an AMM acting as a price oracle.
The attacker bets on a given price, then unbalances the AMM to obtain the desired exchange rate, and finally re-balances the AMM.
In this way, the attacker can win the bet, extracting MEV.
Note that this attack is possible whenever in the state there are enough tokens to unbalance the AMM as required. 
The Bounty contract, which rewards the first user who submits the solution to a puzzle, is a paradigmatic case where a na\"ive implementation leads to ``bad'' MEV attacks that makes the adversary able to steal a submitted solution. 
Fixing the contract requires to devise a non-trivial commit-reveal protocol%
\iftoggle{arxiv}{(as done in~\Cref{ex:bounty})}{}, 
through which we eventually achieve MEV-freedom. 
\section{Related work}
\label{sec:related}

The first (partial) formal definition of MEV was given by
Babel, Daian, Kelkar and Juels~\cite{Babel23clockwork}.
Transliterated into our notation, it is:
\begin{equation}
  \label{eq:mev-babel}
  \mathrm{MEV}^{\textrm{BDKJ}}_{\pmvA}(\sysS) =
  \max \setcomp{\gain{\pmvA}{\sysS}{\TxYS}}{\TxYS \in \mathrm{Blocks}(\pmvA,\sysS)}
\end{equation}
where $\mathrm{Blocks}(\pmvA,\sysS)$ represents the set of all
valid blocks that $\pmvA$ can construct in $\sysS$,
and $\pmvA$ is allowed to own multiple wallets.
A first key difference \wrt our work is that,
while $\mathrm{Blocks}(\pmvA,\sysS)$ is not specified
in~\cite{Babel23clockwork},
we provide an axiomatization of this set in~\Cref{def:mall}.
Notably, our axiomatization allows us to prove key properties of MEV,
as monotonicity (\wrt mempools and wallets) and actors/mempool finiteness.
The dependence of $\mathrm{Blocks}(\pmvA,\sysS)$ on the mempool is
left implicit in~\cite{Babel23clockwork},
while we make the mempool a parameter of MEV.
This allows our attackers to craft transactions by combining their private knowledge with that of the mempool,
as in the attack to the \txcode{BadHTLC} 
\iftoggle{arxiv}{in~\Cref{ex:badhtlc}}{in~\cite{BZ25arxiv}}.
Instead, in \cite{Babel23clockwork} the transactions in
$\mathrm{Blocks}(\pmvA,\sysS)$ are either generated by $\pmvA$
using only her private knowledge, or taken from the (implicit) mempool.


Another key difference is that~\cite{Babel23clockwork} does not provide
a notion of \emph{universal} MEV,
\ie, the MEV that can be extracted by anyone, regardless of their identity
and current token balance.
Indeed, this is the kind of MEV which is most relevant in practice.
The intuition of~\cite{Babel23clockwork} is to compute
$\mathrm{MEV}^{\textrm{BDKJ}}_{\pmvA}(\sysS)$
\wrt an actor $\pmvA$ who is ``external'' to the contract.
However, the intuition is not supported by a formalization,
which is not straightforward to achieve in general.
Instead, our~\Cref{def:adv-mev} exactly characterises
this universal MEV, making it identity-agnostic and token-agnostic.

Before ours, a version of universal MEV was proposed by
Salles~\cite{Salles21formalization},
and it has the following form:
\begin{equation}
  \label{eq:mev-salles}
  \mathrm{MEV}^{\textrm{Salles}}(\sysS) =
  \min_{\pmvA \in \PmvU} \mathrm{MEV}^{\textrm{BDKJ}}_{\pmvA}(\sysS)
\end{equation}
By taking the minimum over all actors, \eqref{eq:mev-salles} no longer depends on the identity of the attacker.
As noted in~\cite{Salles21formalization},
a drawback is that such definition classifies as MEV-free
contracts that intuitively are not:
\eg, the \txcode{Blacklist} contract
would be considered MEV-free, since performing the attack requires
upfront costs, that not all actors can afford.
Note instead that 
\iftoggle{arxiv}{\Cref{th:blacklist} correctly classifies}{we correctly classify} 
\txcode{Blacklist}
as \emph{not} MEV-free.
A fix proposed in~\cite{Salles21formalization,Mazorra22price}
is to parameterise MEV by a constant $n$,
which restricts the set of attackers to those who own at least $n$ tokens:
\begin{equation}
  \label{eq:mev-salles:2}
\mathrm{MEV}^{\textrm{Salles}}(\sysS,n) =
\min \setcomp{\mathrm{MEV}^{\textrm{BDKJ}}_{\pmvA}(\sysS)}{\pmvA \in \PmvU, \WmvA(\pmvA) \geq n}
\end{equation}
where $\WmvA(\pmvA)$ is the number of tokens in $\pmvA$'s wallet.
We note that also this fix has drawbacks.
A first issue is that the set of actors who own at least $n$ tokens in
a state $\sysS$ is always \emph{finite}.
Hence, a contract that blacklists all these actors
preventing them to withdraw tokens
would be MEV-free according to $\mathrm{MEV}^{\textrm{Salles}}(\sysS,n)$. 
Instead, \Cref{def:mev-free} correctly classifies it as \emph{not} MEV-free,
since for all cofinite $\PmvB$, the tokens can be
\emph{redistributed} to some $\PmvA \subseteq \PmvB$
who are not blacklisted, and can then extract MEV.
Another issue of making the $\min$ in~\eqref{eq:mev-salles:2} 
range over all actors is the following.
Consider a variant of \txcode{Blacklist}
where calling \txcode{pay} requires zero tokens.
Since the $\min{}$ in~\eqref{eq:mev-salles:2} must take also
the hard-coded $\pmvA$ into account, and $\pmvA$'s MEV is zero
(since she is blacklisted)
then the $\min$ would be $0$, and so~\eqref{eq:mev-salles:2}
would incorrectly classify the contract as MEV-free.
Instead, our redistribution game allows us to rule out such $\pmvA$.

The notion of MEV in~\cite{BMZ24fc} is based on ours, but it uses an alternative approach to make MEV independent from the wealth of adversaries: rather than using a token redistribution, it takes the maximum MEV over all possible user wallets.
Unlike ours, \cite{BMZ24fc} does not provide a model of the adversarial knowledge. 


Using redistributions like ours is also helpful to
solve another issue of~\cite{Salles21formalization}:
namely, blacklisting could also be based on the number of tokens
held in wallets, \eg, preventing actors with more than 100 tokens from
extracting MEV.
In this case, $\mathrm{MEV}^{\textrm{Salles}}(\sysS,n)$ would be zero for all $n$,
since the minimum would also take into account the blacklisted actors with zero MEV,
while our notion correctly classifies the contract as \emph{not} MEV-free.




As discussed in~\Cref{sec:adv-mev}, our notion of universal MEV
is game-theoretic.
Another approach based on game theory ---
but with substantially different goals ---
is followed in~\cite{Mazorra22price}.
This work models the priority gas auction arising from MEV extraction as a game,
and studies the Nash equilibria ensuring that adversaries have the
same MEV opportunities.
Our goal instead is to formalize MEV so to analyse contracts \wrt MEV attacks.
\section{Limitations and conclusions}
\label{sec:limitations}

While designing our MEV model, we strove to capture 
the most important aspects of MEV in common smart contracts. 
However, our model still has some limitations, which we discuss below.


\paragraph{Long-range attacks} 
Our notion of MEV models the value extractable by adversaries in a single block:
indeed, in \Cref{def:mev} we allow the adversary to perform a sequence $\TxYS$ of transactions (which models the block), but we neglect additional transactions happening after $\TxYS$.
This does not consider long-range attacks spanning across multiple blocks. 
\Eg, the adversary could perform some 
\iftoggle{arxiv}{contract}{} 
actions in one block that do not extract any MEV immediately, but affect the state of a time-based contract that will eventually give MEV.
Precisely addressing long-range attacks would require some extensions to our theory. 
First, the notion of wealth should take into account token price fluctuations, which are irrelevant within a single block.
Second, the knowledge $\mall{}{}$ should also depend on the blockchain state, \eg, to take revealed secrets into account.
Third, MEV should depend on the strategies of the honest actors,
\ie, on the transactions that they would send to the mempool in a given state.
Actually, not considering these strategies, and just assuming that any actor is always willing to perform any contract action, would result in a gross over-approximation of MEV.
An additional complication is that in certain contracts, like \eg, in gambling games, such strategies could be probabilistic.
There, defining MEV in terms of the best possible future state for the adversary would again provide an over-approximation of MEV, since it would assume an unrealistically lucky adversary.
\Eg, consider a guessing game where an actor commits to a secret number, and then the adversary must guess its parity to win. Taking the \emph{maximal} gain over all possible futures effectively provides the adversary with the knowledge of the secret.


\paragraph{Computational adversaries}

Our notion of adversarial knowledge in~\Cref{def:mall} is \emph{sharp}: any piece of data is either known to the adversary or completely inaccessible to them.
This assumption is common in symbolic models of cryptographic protocols, but it does not always perfectly model the real-world adversaries.
Indeed, an adversary could be able to obtain some data but only at the cost of a long computation.
For instance, a contract could require the adversary to solve a moderately hard cryptographic puzzle to extract MEV.
Modeling this kind of computational adversaries would require to refine the notions of adversarial knowledge and MEV to take costs into account. 


\paragraph{Cost of MEV}

Our notion of universal MEV evaluates the MEV that can be extracted by an arbitrarily wealthy adversary. 
%
To this purpose, \Cref{def:adv-mev} uses token redistributions, which allow the adversary to use \emph{all}  tokens in the state, even those belonging to honest actors.
In this way, we are effectively assuming that the adversary always has, either in their wallet or by buying them from honest actors,
all the tokens needed to carry the attack.
An alternative definition, which does not require to grab the tokens of honest actors, would be to allow the adversary to mint the tokens needed in the attack, similarly to the definition of 
${\rm MEV}^{\infty}$ in~\cite{BMZ24fc}.
In scenarios where tokens are used as credentials to perform given actions, we could restrict token redistribution to avoid giving such special tokens to the adversary.
We also remark that transaction fees do not contribute to MEV,
similarly to~\cite{Babel23clockwork,Salles21formalization}.
Encompassing fees would allow to declassify as MEV attacks
those where the fees needed to carry the attack exceed the
adversarial gain.
In practice, fees may be quite costly in private mempools
like Flashbots~\cite{Weintraub22flashbot}.
Adversaries can extract MEV also
by front-running a transaction in the mempool
so to increase the amount of gas needed to validate it.


\paragraph{Good \emph{vs.} bad MEV}

There is an open debate 
\iftoggle{arxiv}{within the community}{} 
about what exactly constitutes MEV, and how to separate ``good MEV'' from ``bad MEV''~\cite{flashbots-mev-taxonomy,Barczentewicz23ssrn,Ji24fc,Torres24ccs,monoceros-mev,a-new-game-in-town}.
In the absence of an agreement about these notions, our
definition of ``bad MEV'' in~\Cref{def:badmev} formally captures some of the arguments used in this debate. 
In particular, we classify as ``good'' the MEV obtained 
through arbitrage in AMMs, since it does not exploit the mempool.
Furthermore, we classify the liquidations on Lending Pools obtained without exploiting the mempool as ``good'' MEV.
These classifications seem coherent with discussions in the community: MEV is good when it is an incentive for \emph{any} user to perform actions (\eg, arbitrage, liquidations) that serve to the purpose of the protocol (\eg, aligning prices for AMMs, repaying loans for Lending Pools).
Instead, we classify as ``bad'' the MEV resulting from
sandwich attacks to AMMs~\cite{Zhou21high} and from 
liquidations that back-run interest-accruing transactions in Lending Pools: this is correct in our view, because these attacks require the privileges of block proposers (by contrast, plain arbitrages and liquidations can be performed by any user with sufficient tokens).
We stress that not all the intuitively bad MEVs are classified as such by our~\Cref{def:badmev}:
this is the case, \eg, of the DAO attack~\cite{DAO}, where MEV results from a bug in the contract implementation.
On the other side, we are not aware of any real-world contracts that have ``bad MEV'' (according to our definition), but where the MEV is considered beneficial to the contract functionality.

\paragraph{Private order flows}
The ``no shared secret'' axiom in~\Cref{def:mall} forbids actors to share private information (\eg, keys). 
It also rules out \emph{private order flows}~\cite{Gupta23aft}: if the same order flow is sent to $\pmvA$ and $\pmvB$, they might infer some common transactions that are not public knowledge.
This simplifying assumption can be relaxed by making variables $\pmvA$ represent secrets rather than actors, and modelling actors as the set of secrets $\PmvA$ that they know.


\iftoggle{anonymous}{}{\paragraph*{Acknowledgments}

Work partially supported by the MUR National Recovery and Resilience Plan funded by the European Union -- NextGenerationEU,
projects SERICS (PE00000014) and PRIN 2022 DeLiCE (F53D23009130001).}

\bibliographystyle{splncs04}
\bibliography{main}

\iftoggle{arxiv}{%
\clearpage
\appendix

\section{A concrete contract language}
\label{sec:txscript}

We instantiate our abstract model 
with a simple contract language, dubbed \txscript.
The language is heavily inspired by Solidity --- and indeed we could have
as well used Solidity itself to sketch our examples.
Still, we opt for introducing a new language for coherence with
the spirit of this paper:
to formally prove properties of contracts, \eg, the presence or absence of MEV,
we need a language with a formal semantics.
To ease formal definitions and reasoning, our \txscript aims at minimality,
so we drop all the features of Solidity that are inessential to
the understanding of MEV.
Despite this simplification, \txscript is still expressive enough
to express real-world use cases like those found in DeFi.

A \txscript contract is a finite set of procedures of the form:
\[
\procF(\overrightarrow{\textit{parg}}) \braceleft s \braceright
\]
where $\procF$ is the procedure name,
$\overrightarrow{\textit{parg}}$ is the sequence of formal parameters,
and $s$ is the procedure body
(see \Cref{fig:txscript:syntax}).
We assume that all the procedures in a contract have distinct names.
Statements and expressions extend those of a loop-free imperative language
with a few domain-specific constructs:
\begin{itemize}

\item $\balance{\tokT}$ is the balance of tokens of type $\tokT$ deposited in the contract;

\item $\txout{\pmvA}{n}{\tokT}$ transfers $n$ units of $\tokT$ from the contract to $\pmvA$;

\item $\reqE{e}$ rolls-back the transaction if condition $e$ is false.

\end{itemize}

Transactions have the form $\procF(\overrightarrow{\textit{txarg}})$,
where the sequence of \emph{actual} parameters may include, besides constants,
the term $\txin{\pmvA}{n}{\tokT}$, representing
a transfer of $n$ units of token $\tokT$
from $\pmvA$ to the contract upon the procedure call.%
\footnote{This mechanism generalises the one provided by Ethereum to transfer tokens upon contract calls.
  In Ethereum, a contract call involves a single transfer of \emph{ether} from the caller to the contract.
  In \txscript, instead, a single transaction can involve multiple transfers
  of tokens (of any type) from the actors who authorise the transaction.
}
Note that $\txin{\pmvA}{n}{\tokT}$ involves a signature of $\pmvA$ on the transaction that authorize the transfer of tokens from 
her wallet to the contract.
When we are only interested in the signature, and not in the
transfer of tokens, we just write $\txsig{\pmvA}$, interpreting it as syntactic sugar for $\txin{\pmvA}{0}{\tokT}$.
Transactions nonces $@n$ are used to prevent double-spending attacks:
all transactions in the blockchain must have distinct nonces.
We omit transaction nonces in examples.





\begin{figure}[htbp]
  \centering
  \fbox{
    \begin{minipage}{0.925\linewidth}    
  \[
  \small
  \begin{array}{rcll}
    c & ::= & \txscriptcode{contract} \; \contrC \braceleft \; \overrightarrow{\textit{p}} \; \braceright & \textbf{Contract}
    \\[5pt]
    \textit{p} & ::= & \procF(\overrightarrow{\textit{parg}}) \braceleft s \braceright & \textbf{Procedure}
    \\[5pt]
    \textit{parg} & ::= & & \textbf{Argument}
    \\
    & & x & \text{variable}
    \\
    & \mid & \txin{\pmva}{x}{\tokt} & \text{token input}
    \\
    & \mid & \txsig{\pmva} & \text{signature}
    \\[5pt]    
    s & ::= & & \textbf{Statement}
    \\
    & & \skipE & \text{skip}
    \\
    & & \reqE e & \text{require condition}
    \\    
    & \mid & x = e & \text{assignment}
    \\
    & \mid & x[e_1] = e_2 & \text{map update}
    \\
    & \mid & \txout{e_1}{e_2}{e_3} & \text{token output}
    \\
    & \mid & s_1 ; s_2 & \text{sequence}
    \\
    & \mid & \ifE{e}{s_1}{s_2}\qquad & \text{conditional}
    \\[5pt]
    e & ::= & & \textbf{Expression}
    \\
    & & \nullE & \text{undefined}
    \\
    & \mid & n \mid \pmvA \mid \tokT & \text{constants}
    \\
    & \mid & x & \text{variables}
    \\
    & \mid & e_1[e_2] & \text{map lookup}
    \\    
    & \mid & \op{\circ}{e_1}{e_2} & \text{operation}
    \\
    & \mid & \balance{e} & \text{number of tokens of type $e$}
    \\
    & \mid & \hashE{e_1,\ldots,e_n} & \text{collision-resistant hash}
    \\[5pt]
    \txT & ::= & \procF(\overrightarrow{\textit{txarg}})@n & \textbf{Transaction}
    \\[5pt]
    \textit{txarg} & ::= & & \textbf{Transaction argument}
    \\
    & & n \mid \pmvA \mid \tokT & \text{constants}
    \\
    & \mid & \txin{\pmvA}{n}{\tokT} & \text{token input}
    \\
    & \mid & \txsig{\pmvA} & \text{signature}
  \end{array}
  \]
    \end{minipage}
  }
  \caption{Syntax of contracts and transactions.}
  \label{fig:txscript:syntax}
\end{figure}

To improve readability,
when we want to fix some parameters in a procedure, instead of writing:
\begin{lstlisting}[language=txscript,morekeywords={HTLC,commit,reveal,timeout},classoffset=3,morekeywords={a,b,o,A,B,Oracle},keywordstyle=\pmvColor,classoffset=4,morekeywords={t,BTC},keywordstyle=\tokColor,basicstyle=\fontseries{m}\footnotesize\ttfamily]
commit(a pays x:t,b,c) { 
    require a==A && x==1 && t==BTC && b==B; ... 
}
\end{lstlisting}
we just hard-code constants in the formal parameters, \eg:
\begin{lstlisting}[language=txscript,morekeywords={HTLC,commit,reveal,timeout},classoffset=3,morekeywords={a,b,o,A,B,Oracle},keywordstyle=\pmvColor,classoffset=4,morekeywords={t,BTC},keywordstyle=\tokColor,basicstyle=\fontseries{m}\footnotesize\ttfamily,]
    commit(A pays 1:BTC,B,c) { ... }
\end{lstlisting}


\begin{example}
  \label{ex:txscript:htlc:sem}
  Recall the \txcode{HTLC} from~\Cref{fig:htlc},
  and let:
  \[
  \walu{\pmvA}{1}{\tokT} \mid
  \walu{\pmvB}{0}{\tokT} \mid
  \walpmv{\contract{HTLC}}{\sigma_0}
  \]
  be an initial state where $\pmvA$ owns \mbox{$\waltok{1}{\tokT}$},
  $\pmvB$ owns nothing, and \contract{HTLC} is in the initial state
  $\sigma_0$ where the balance is empty and all the variables
  are set to their default values (as in Solidity).
  Assume that $\PmvA$ chooses a secret $s$, and computes
  its hash $\hashSem{s} = h$.
  Upon firing
  $\txcode{commit}(\txin{\pmvA}{1}{\tokT},\pmvB,h)$  
  the state takes a transition to:
  \[
  \walu{\pmvA}{0}{\tokT} \mid
  \walu{\pmvB}{0}{\tokT} \mid
  \walpmv{\contract{HTLC}}{\sigma_0\setenum{\bind{\code{xa}}{\pmvA}, \bind{\code{xb}}{\pmvB}, \bind{\code{yc}}{h}}; \waltok{1}{\tokT}}  
  \]
  Now, upon firing $\txcode{reveal}(\rk{\nonceR[\pmvA]}{b})$,
  the state evolves to:
  \[
  \walu{\pmvA}{1}{\tokT} \mid
  \walu{\pmvB}{0}{\tokT} \mid
  \walpmv{\contract{HTLC}}{\sigma_0\setenum{\bind{\code{xa}}{\pmvA}, \bind{\code{xb}}{\pmvB}, \bind{\code{yc}}{h}}; \waltok{0}{\tokT}}
  \]
  In this state, the committed secret has been revealed,
  and the committer $\pmvA$ has redeemed her deposit.
  \hfill\qedex
\end{example}

\paragraph{Semantics}

We assume a set $\VarU$ of variables (ranged over by $x, y, \ldots$),
partitioned in two subsets:
$\VarUP$ for the variables used in procedure arguments,
and $\VarUS$ for the state variables,
which are further partitioned into base variables and map variables.
We use $\tokt$ for variable token types,
and $\pmva$ for variable actors.
We assume that the variables used in the LHS of any assignment and in
map updates are in $\VarUS$.

Contract states are pairs of the form $(\sigma,\TxT)$.
The component $\sigma$ is a total map in
\mbox{$\VarUS \cup \TokU \rightarrow \ValU$},
where $\ValU$ is the universe of values,
comprising base values $\ValUB$, and total maps from base values to values.
Base values are natural numbers ($\Nat$),
actors ($\PmvU$),
tokens ($\TokU$),
and the singleton $\setenum{\nullSem}$.
We embed booleans into $\Nat$ as usual.

The component $\TxT$ of the contract state records the set of all
transactions executed so far, and it is used to prevent the double-spending
of transactions in the mempool.%
\footnote{Blockchain platforms use similar mechanisms
  to avoid replay attacks which double spend a transaction.
  For instance, Algorand marks a transaction
  as invalid if it belongs to the set of transactions fired in the last
  1000 rounds.
  In Ethereum, each transaction must be signed by its sender:
  the signature also includes a nonce, which is increased each time the
  sender broadcasts a transaction.
  In the blockchain, the nonces of the transactions from the same sender must
  be ordered, without skipping.}
A contract state $(\sigma,\TxT)$ is \emph{initial} when $\sigma(\tokT) = 0$
for all $\tokT \in \TokU$,
$\sigma(x) = \nullSem$ for all base variables,
$\sigma(x) = \lambda b.\, \nullSem$ for all map variables,
and $\TxT = \emptyset$.

We will often omit transaction nonces:
if the same procedure with the same parameters is executed two or more times,
we implicitly assume that all its transaction nonces are distinct.

The semantics of contracts is a labelled transition relation
between blockchain states, with signature as in~\Cref{def:contract}.
The transition relation $\xmapsto{}$ is specified by the following rule,
which updates the blockchain state when a valid transaction is fired:
\[
  \irule
      {\begin{array}{c}
          \txT = \procF(\overrightarrow{\textit{txarg}})@n
          \qquad
          \procF(\overrightarrow{\textit{parg}}) \braceleft s \braceright \in \contrC
          \qquad
          \txT \not\in \TxT         
          \\[8pt]
          (\overrightarrow{\textit{parg}}) \rho = \overrightarrow{\textit{txarg}}
          \qquad
          \semcmd{\overrightarrow{\textit{txarg}}}{(\WmvA,\sigma)} \Rightarrow_{\it arg} (\WmvAii,\sigma'')
          \qquad
          \semcmd[\rho]{s}{(\WmvAii,\sigma'')} \Rightarrow (\WmvAi,\sigma')
      \end{array}}
      {(\WmvA,(\sigma,\TxT)) \xmapsto{\txT} (\WmvAi,(\sigma',\TxT \cup \setenum{\txT}))}
\]
The rule defines a single state transition
triggered by a (valid) transaction $\txT$.
The condition $\txT \not\in \TxT$ in the first line of the rule premises
ensures that the same transaction cannot be executed twice.
The second line of the rule premises infers a substitution $\rho$
to match the formal and the actual parameters of the called procedure.
The condition
$\semcmd{\overrightarrow{\textit{txarg}}}{(\WmvA,\sigma)} \Rightarrow_{\it arg} (\WmvAii,\sigma'')$
evaluates the transaction arguments (see below).
Finally, the premise
$\semcmd[\rho]{s}{(\WmvAii,\sigma'')} \Rightarrow (\WmvAi,\sigma')$
evaluates the procedure statement $s$,
producing a new blockchain state.
Note that if some $\reqE{}$ commands in the statement $s$ fail
then $\semcmd[\rho]{s}{(\WmvAii,\sigma'')} \Rightarrow \bot$,
hence the premise is false, and the rule does not apply.

The semantics of expressions in a state $\sigma$ is standard,
except for the wallet lookup, the semantics of which
is defined as follows:
\[
\begin{array}{c}
  \irule{\semexp[\rho,\sigma]{e} = \tokT}
        {\semexp[\rho,\sigma]{\balance{e}} = \sigma(\tokT)}
        
\end{array}
\]


We assume a basic type system on expressions, which rules out
operations between non-compatible types, like \eg, 
$\pmvA + 1$, $\tokT + 1$, $\balance{\pmvA}$, $\secE{e}$ where $e$
is anything but a reveal $\rk{r}{n}$.

The rule for transferring tokens from the contract to an actor
is the following,
where we use the standard notation $\sigma \setenum{\bind{x}{v}}$
to update a partial map $\sigma$ at point $x$:
namely, $\sigma \setenum{\bind{x}{v}}(x) = v$, while
$\sigma \setenum{\bind{x}{v}}(y) = \sigma(y)$ for $y \neq x$.
\[
\irule
    {\begin{array}{c}
        \semexp[\rho,\sigma]{e_1} = \pmvA \quad
        \semexp[\rho,\sigma]{e_2} = n \quad
        \semexp[\rho,\sigma]{e_3} = \tokT \quad
        \sigma(\tokT) \geq n
      \end{array}
    }
    {\semcmd[\rho]{\txout{e_1}{e_2}{e_3}}{(\WmvA,\sigma)} \Rightarrow
      (\WmvA\setenum{\bind{\pmvA}{\WmvA(\pmvA) + \walenum{\waltok{n\,}{\,\tokT}}}},\sigma - \walenum{\waltok{n}{\tokT}})}
\]

The rules for evaluating $\reqE$ commands are the following:
\[
\irule
    {\semexp[\rho,\sigma]{e} = \true}
    {\semcmd[\rho]{\reqE e}{(\WmvA,\sigma)} \Rightarrow (\WmvA,\sigma)}
\qquad
\irule
    {\semexp[\rho,\sigma]{e} = \false}
    {\semcmd[\rho]{\reqE e}{(\WmvA,\sigma)} \Rightarrow
      \bot}
\]
The second rule deal with the case when the condition under the $\reqE$
is violated: in this case, the evaluation of the command yields the
special value $\bot$, which represents an execution error.

The rule for evaluating actual parameters is the following:
\[
\irule
    {\begin{array}{c}
        \WmvA(\pmvA)(\tokT) \geq n
        \qquad
        \semcmd{\overrightarrow{\textit{txarg}}}{(\WmvA\setenum{\bind{\pmvA}{\WmvA(\pmvA) - \walenum{\waltok{n\,}{\,\tokT}}}},\sigma + \walenum{\waltok{n}{\tokT}})} \Rightarrow_{\it arg}
      (\WmvAi,\sigma')
      \end{array}
    }
    {\semcmd{\txin{\pmvA}{n}{\tokT}; \overrightarrow{\textit{txarg}}}{(\WmvA,\sigma)} \Rightarrow_{\it arg}
      (\WmvAi,\sigma')}
\]
The other rules are standard.
The full set of rules is in~\Cref{fig:txscript:sem}.

\begin{figure}[htbp]
  \centering
  \small
  \fbox{
  \begin{tabular}{c}
  $
      \begin{array}{c}
      \\[0pt]
        \irule{}
              {\semexp[\rho,\sigma]{\nullE} = \nullSem}
        \qquad
        \irule{}
              {\semexp[\rho,\sigma]{v} = v}
        \qquad
        \irule{}
              {\semexp[\rho,\sigma]{x} = (\rho \cup \sigma)(x)}
        \qquad
        \irule{\semexp[\rho,\sigma]{e} = \tokT}
              {\semexp[\rho,\sigma]{\balance{e}} = \sigma(\tokT)}
        \\[15pt]
        \irule{\semexp[\rho,\sigma]{e_1} = f \in \ValUB \rightarrow \ValU \quad
          \semexp[\rho,\sigma]{e_2} = v \in \ValUB
        }
        {\semexp[\rho,\sigma]{e_1[e_2]} = f(v)}       
        \\[15pt]
        \irule{\semexp[\rho,\sigma]{e_1} = n_1 \quad
          \semexp[\rho,\sigma]{e_2} = n_2 \quad n_1 \circ n_2 = n \in \Nat}
              {\semexp[\rho,\sigma]{e_1 \circ e_2} = n}
        \\[15pt] 
        \irule
            {\semexp[\rho,\sigma]{e_1} = v_1 \quad \cdots \quad
              \semexp[\rho,\sigma]{e_n} = v_n \quad
              \hashSem{v_1 \cdots v_n} = v}
            {\semexp[\rho,\sigma]{\hashE{e_1,\ldots,e_n}} = v}             
   \end{array}
   $
   \\[70pt]
   $
    \begin{array}{c}
        \irule
            {}
            {\semcmd[\rho]{\skipE}{(\WmvA,\sigma)} \Rightarrow (\WmvA,\sigma)}    
        \\[15pt]
        \irule
            {\semexp[\rho,\sigma]{e} = \true}
            {\semcmd[\rho]{\reqE e}{(\WmvA,\sigma)} \Rightarrow
              (\WmvA,\sigma)}    
        \qquad
        \irule
            {\semexp[\rho,\sigma]{e} = \false}
            {\semcmd[\rho]{\reqE e}{(\WmvA,\sigma)} \Rightarrow
              \bot}    
        \\[15pt]           
        \irule
            {\semexp[\rho,\sigma]{e} = v}
            {\semcmd[\rho]{x = e}{(\WmvA,\sigma)} \Rightarrow
              (\WmvA,\sigma\setenum{\bind{x}{v}})}    
        \\[15pt]
        \irule
            {\begin{array}{c}
                \semexp[\rho,\sigma]{x} = f \quad          
                \semexp[\rho,\sigma]{e_1} = k \quad
                \semexp[\rho,\sigma]{e_2} = v \quad
                f' = f\setenum{\bind{k}{v}}
              \end{array}
            }
            {\semcmd[\rho]{x[e_1] = e_2}{(\WmvA,\sigma)} \Rightarrow
              (\WmvA,\sigma\setenum{\bind{x}{f'}})}
        \\[15pt]
        \irule
            {\begin{array}{c}
                \semexp[\rho,\sigma]{e_1} = \pmvA \quad
                \semexp[\rho,\sigma]{e_2} = n \quad
                \semexp[\rho,\sigma]{e_3} = \tokT \quad
                \sigma(\tokT) \geq n
              \end{array}
            }
            {\semcmd[\rho]{\txout{e_1}{e_2}{e_3}}{(\WmvA,\sigma)} \Rightarrow
              (\WmvA\setenum{\bind{\pmvA}{\WmvA(\pmvA) + \waltok{n\,}{\,\tokT}}},\sigma - \waltok{n}{\tokT})}
        \\[15pt]
        \irule
            {\begin{array}{c}
                \semcmd[\rho]{s_1}{(\WmvA,\sigma)} \Rightarrow (\WmvA'',\sigma'')
                \quad
                \semcmd[\rho]{s_2}{(\WmvA'',\sigma'')} \Rightarrow (\WmvA',\sigma')
              \end{array}
            }
            {\semcmd[\rho]{s_1;s_2}{(\WmvA,\sigma)} \Rightarrow
              (\WmvA\setenum{\bind{\pmvA}{\WmvA(\pmvA) + \waltok{n\,}{\,\tokT}}},\sigma - \waltok{n}{\tokT})}
        \\[15pt]
        \irule
            {\begin{array}{c}
                \semexp[\rho,\sigma]{e} = b \quad
                \semcmd[\rho]{s_b}{(\WmvA,\sigma)} \Rightarrow (\WmvA',\sigma')
              \end{array}
            }
            {\semcmd[\rho]{\ifE{e}{s_{\true}}{s_{\false}}}{(\WmvA,\sigma)} \Rightarrow (\WmvA',\sigma')}
    \end{array}
    $
    \\[120pt]
    $
      \begin{array}{c}
        \irule
            {}
            {\semcmd{\varepsilon}{(\WmvA,\sigma)} \Rightarrow_{\it arg} (\WmvA,\sigma)}
        \qquad
        \irule
            {\semcmd{\overrightarrow{\textit{txarg}}}{(\WmvA,\sigma)} \Rightarrow_{\it arg} (\WmvAi,\sigma')}
            {\semcmd{v ; \overrightarrow{\textit{txarg}}}{(\WmvA,\sigma)} \Rightarrow_{\it arg} (\WmvAi,\sigma')}
        \\[20pt]
        \irule
            {\begin{array}{c}
                \WmvA(\pmvA)(\tokT) \geq n
                \quad
                \semcmd{\overrightarrow{\textit{txarg}}}{(\WmvA\setenum{\bind{\pmvA}{\WmvA(\pmvA) - \waltok{n\,}{\,\tokT}}},\sigma + \waltok{n}{\tokT})} \Rightarrow_{\it arg}
                (\WmvAi,\sigma')
                \\[2pt]
              \end{array}
            }
            {\semcmd{\txin{\pmvA}{n}{\tokT}; \overrightarrow{\textit{txarg}}}{(\WmvA,\sigma)} \Rightarrow_{\it arg}
              (\WmvAi,\sigma')}
    \\[30pt]
            \irule
            {\begin{array}{c}
                \txT = \procF(\overrightarrow{\textit{txarg}})@n
                \qquad
                \procF(\overrightarrow{\textit{parg}}) \braceleft s \braceright \in \contrC
                \qquad
                \txT \not\in \TxT         
                \\[5pt]
                (\overrightarrow{\textit{parg}}) \rho = \overrightarrow{\textit{txarg}}
                \qquad
                \semcmd{\overrightarrow{\textit{txarg}}}{(\WmvA,\sigma)} \Rightarrow_{\it arg} (\WmvAii,\sigma'')
                \qquad
                \semcmd[\rho]{s}{(\WmvAii,\sigma'')} \Rightarrow (\WmvAi,\sigma')
                \\[2pt]
            \end{array}}
            {(\WmvA,(\sigma,\TxT)) \xmapsto{\txT} (\WmvAi,(\sigma',\TxT \cup \setenum{\txT}))}
            \\[20pt]
      \end{array}
    $
    \end{tabular}
   } 
   \caption{Semantics of contracts.}
   \label{fig:txscript:sem}
\end{figure}

\section{Evaluation}
\label{sec:examples}

We now assess the effectiveness of our MEV theory
on a benchmark of real-world contracts.

\subsection{Whitelist}
\label{ex:whitelist}

\begin{thm}
  \txcode{Whitelist} is $\mev{}{}{}$-free, for all $\sysS$ and~$\TxT$. 
\end{thm}
\begin{proof}
  Let $\sysS$ be such that \txcode{Whitelist}
  contains $\waltok{n}{\tokT}$ with $n>0$.
  Let $\TxT$ be arbitrary, and
  let $\PmvB$ be a cofinite set \emph{not} including $\pmvA$.
  For all $\PmvA \subseteq \PmvB$ and for all redistributions
  $\waldistr{}{\sysS}{}{\sysSi}$,
  we have that $\mev{\PmvA}{\sysSi}{\TxT} = 0$,
  hence the $\max{}$ in~\eqref{eq:adv-mev} is zero.
  Therefore, $\mev{}{\sysS}{\TxT} = 0$:
  indeed, any reachable state is MEV-free in any mempool.
  Note that taking the $\min{}$ \wrt all cofinite $\PmvB$ in \Cref{def:adv-mev}
  is instrumental to exclude from the potential adversaries those actors
  which are assigned a privileged role by the contract, as 
  $\pmvA$ in \txcode{Whitelist}.
  In this way, adversaries cannot exploit their identities to extract MEV.
\end{proof}

\subsection{Blacklist}
\label{ex:blacklist}

\begin{thm}
  \label{th:blacklist}
  \txcode{Blacklist} is \emph{not} $\mev{}{}{}$-free,
  for all $\TxT$ and
  all $\sysS$ where both wallets and the contract
  contain at least $\waltok{1}{\tokT}$.
\end{thm}
\begin{proof}
  Let $\PmvB$ be any cofinite set of actors,
  let $\TxT$ be an arbitrary mempool,
  and assume that \contract{Blacklist} contains
  $\waltok{n}{\tokT}$ with $n>0$.  
  Let $\PmvA = \PmvB$, and let
  $\waldistr{}{\sysS}{}{\sysSi}$
  assign at least $\waltok{1}{\tokT}$ to some $\pmvM \neq \pmvA$ in $\PmvA$.
  This maximizes $\mev{\PmvA}{\sysSi}{\TxT} = n \cdot \wealth{}{}\idxfun{\tokT}$. 
  Therefore, taking the $\min{}$ \wrt all cofinite sets of actors
  yields $\mev{}{\sysS}{\TxT} = n \cdot \wealth{}{}\idxfun{\tokT}$,
  and so the contract is \emph{not} MEV-free, as expected.
  Note that some redistributions would lead to zero MEV:
  \eg, this is the case when all tokens are assigned to the blacklisted $\pmvA$.
  To avoid this issue,
  the $\max{}$ in~\eqref{eq:adv-mev} allows to consider the MEV
  resulting from the most favourable redistribution for the adversary.
  Note that computing the $\min{}$ \wrt all \emph{cofinite} $\PmvB$ 
  ensures that we can always assign tokens to non-blacklisted actors.
\end{proof}

\subsection{Bank}
\label{ex:bank}

\begin{thm}
  \label{ex:adv-mev:bank}
  \txcode{Bank} is $\mev{}{}{}$-free, for all $\sysS$ and finite $\TxT$.
\end{thm}
\begin{proof}
  Let $\PmvB$ be cofinite and not including any actor in $\TxT$
  or in the contract state (which can only mention finitely many actors).
  Observe that $\mev{\PmvA}{\sysSi}{\TxT} = 0$,
  for any $\PmvA \subseteq \PmvB$
  and token redistribution $\waldistr{}{\sysS}{}{\sysSi}$.
  Indeed, there are only two ways for $\PmvA$
  to extract tokens from the contract:
  via a \txcode{wdraw}, or a \txcode{xfer} followed by a \txcode{wdraw}.
  Now, since the contract state does not mention actors in $\PmvA$,
  \txcode{wdraw} transactions fired from $\PmvA$ are invalid
  (unless they are preceded by a \txcode{deposit},
  but this would lead to a non-positive gain).
  Since the mempool $\TxT$ does not mention actors in $\PmvB$,
  \txcode{xfer} transactions to actors in $\PmvA \subseteq \PmvB$
  are not forgeable by $\PmvA$.
  Therefore, the $\max{}$ in~\eqref{eq:adv-mev} is zero,
  and so any contract state is MEV-free.
\end{proof}

%
%

\subsection{HTLC}
\label{ex:htlc}

\label{ex:txscript:badhtlc:mall}
The \txcode{BadHTLC} contract in~\Cref{fig:badhtlc}
implements a Hash-Time Locked Contract, where
a committer promises that she will either reveal a secret
within a certain deadline, or pay a penalty of \mbox{$\waltok{1}{\tokT}$} to anyone after the deadline.
The procedure \txcode{commit} initialises the contract state,
setting the variable \code{commitment}.
The parameter ``$\txin{\pmv{a}}{1}{\tokT}$''
asks any $\pmv{a}$ (who becomes the committer) to deposit \mbox{$\waltok{1}{\tokT}$} into the contract.
%
The procedure \txcode{reveal} allows anyone
to redeem the deposit by revealing the secret: there,
the parameter ``$\txsig{\pmv{a}}$'' requires the transaction
to be signed by $\pmv{a}$;
the command $\txout{\pmv{a}}{\balance{\tokT}}{\tokT}$
transfers to $\pmv{a}$ all the tokens $\tokT$ stored in the contract.
Dually, \txcode{timeout} allows anyone to redeem the deposit
after the deadline, triggered by a time oracle who
signs the transaction.

We anticipate that \txcode{BadHTLC} suffers from a MEV attack
in a state where the secret has been committed but not revealed yet, like \eg:
\begin{equation}
\label{eq:txscript:badhtlc:mall}
\sysS
\; = \;
\walu{\pmvA}{0}{\tokT}
\mid
\walpmv{\contract{BadHTLC}}{\waltok{1}{\tokT}, \code{commitment}=H(s)}
\mid
\cdots
\end{equation}
In state $\sysS$, $\pmvA$ has \mbox{$\waltok{0}{\tokT}$}
and \contract{BadHTLC} has \mbox{$\waltok{1}{\tokT}$}.
Suppose the mempool $\TxT$ contains a transaction
\mbox{$\txT[\pmvA] = \txcode{reveal}(\txsig{\pmvA},s)$}
sent by $\pmvA$ to redeem the deposit.
Since the secret $s$ is public in the mempool,
any adversary $\pmvM$ can craft a transaction
$\txT[\pmvM] = \txcode{reveal}(\txsig{\pmvM},s)$
by combining their own knowledge (to provide $\pmvM$'s signature)
with that of $\TxT$ (to provide~$s$).
Therefore, $\txT[\pmvM] \in \mall{\setenum{\pmvM}}{\TxT}$, and
so $\pmvM$ can extract MEV
by front-running $\txT[\pmvA]$ with $\txT[\pmvM]$.
Indeed, we have that:
\[
\sysS
\xrightarrow{\txT[\pmvM]}
\walu{\pmvA}{0}{\tokT} \mid
\walu{\pmvM}{1}{\tokT} \mid
\walpmv{\contract{BadHTLC}}{\waltok{0}{\tokT}, \cdots}
\mid
\cdots
\]
Note instead that $\pmvM$ \emph{alone} cannot deduce the transaction that would allow her to trigger the timeout. 
Formally,
\(
\txY[\pmvM]
=
\txcode{timeout}(\txsig{\pmvM},\txsig{\pmv{Oracle}})
\not\in \mall{\setenum{\pmvM}}{\TxT}
\).
Crafting $\txY[\pmvM]$ is not possible without the cooperation
of the oracle, even if $\pmvM$ knows a transaction
\(
\txcode{timeout}(\txsig{\pmvB},\txsig{\pmv{Oracle}})
\)
from a past interaction, 
since each signature is tied to a specific transaction.
%
We remark that this attack requires the adversary to combine
private and mempool knowledge,
which does not seem to be properly accounted for in
current MEV formalizations~\cite{Babel23clockwork,Salles21formalization,Mazorra22price}.

\medskip
The \txcode{HTLC} contract in~\Cref{fig:htlc} implements a fix to the \contract{BadHTLC} in~\Cref{fig:badhtlc}.
The fix consists in constraining the reveal method to only transfer tokens to the committer.
The fixed contract is MEV-free in any state and finite mempool.

\begin{figure}[t]
  \begin{lstlisting}[language=txscript,morekeywords={HTLC,commit,reveal,timeout},classoffset=3,morekeywords={a,b,A,committer,verifier,Oracle},keywordstyle=\pmvColor,classoffset=4,morekeywords={t,T},keywordstyle=\tokColor,frame=single]
contract HTLC {
  commit(a pays 1:T,b,c) { // a must send 1 token T to the contract
    require verifier==null; committer=a; verifier=b; commitment=c;
  }
  reveal(a sig,y) { // a must sign the transaction
    require a==committer && balance(T)>0 && H(y)==commitment;
    transfer(a,balance(T):T) // send the whole balance of tokens T to a
  }
  timeout(Oracle sig) { // Oracle must sign the transaction
    require balance(T)>0; // the contract must have some tokens T
    transfer(verifier,balance(T):T); 
    verifier=null; 
  }
}
  \end{lstlisting}
  \negcaptionspace  
  \caption{A Hash Time Locked Contract.}
  \label{fig:htlc}
\end{figure}

%
%

\subsection{Coin Pusher}
\label{ex:coinpusher}

The following example shows the case where extracting MEV
requires to fire some transactions found in the mempool.
The \txcode{CoinPusher} contract in~\Cref{fig:coinpusher}
transfers all its tokens to anyone who
makes the balance exceed $\waltok{100}{\tokT}$.
Let $\sysS$ be a state where the contract has $\waltok{0}{\tokT}$,
$\pmvA$ has $\waltok{1}{\tokT}$, and
$\pmvM$ has $\waltok{99}{\tokT}$.
In the empty mempool, $\pmvM$ has no MEV in $\sysS$,
since she has not enough balance to trigger the push of tokens from the contract.
Instead, with a mempool
$\TxT = \setenum{\txcode{play}(\txin{\pmvA}{1}{\tokT})}$,
we have that $\pmvM$ can fire the sequence
$\txcode{play}(\txin{\pmvA}{1}{\tokT}) \, \txcode{play}(\txin{\pmvM}{99}{\tokT}) \in \mall{\pmvM}{\TxT}^*$,
obtaining $\mev{\setenum{\pmvM}}{\sysS}{\TxT} = 1 \cdot \idxfun{\tokT}$.

\begin{figure}
  \begin{lstlisting}[language=txscript,morekeywords={CoinPusher,play},classoffset=3,morekeywords={a,b,A,Oracle},keywordstyle=\pmvColor,classoffset=4,morekeywords={t,T},keywordstyle=\tokColor,frame=single]
contract CoinPusher {
  play(a pays x:T) {
    require x>0;
    if balance(T)>=100 then transfer(a,balance(T):T) else skip;
  }
}
  \end{lstlisting}
  \negcaptionspace
  \caption{A \txcode{CoinPusher} contract.}
  \label{fig:coinpusher}
\end{figure}

%
%

\subsection{Crowdfund}
\label{ex:crowdfund}

\begin{figure}
  \begin{lstlisting}[language=txscript,morekeywords={Crowdfund,init,donate,claim,timeout,refund},classoffset=3,morekeywords={a,b,o,A,Oracle,rcv},keywordstyle=\pmvColor,classoffset=4,morekeywords={t,T},keywordstyle=\tokColor,frame=single]
contract Crowdfund {
  init(a,n) {
    require rcv==null;
    rcv=a; goal=n; isOpen=true;
  }
  donate(a pays x:T) {
    require isOpen && x>0;
    if amount[a]=null then amount[a]=x else amount[a]=amount[a]+x;
  }  
  claim() {
    require balance(T)>=goal && isOpen;
    transfer(rcv,balance(T):T); rcv=null; isOpen=false;
  }
  timeout(Oracle sig) {
    isOpen=false;
  }
  refund(a sig) {
    require amount[a]>0 && !isOpen;
    transfer(a,amount[a]:T); amount[a]=0;
  }
}
  \end{lstlisting}
  \negcaptionspace  
  \caption{A crowdfunding contract.}
  \label{fig:crowdfund}
\end{figure}

Consider the \txcode{Crowdfund} contract in~\Cref{fig:crowdfund}.
Let $\sysS$ be a state where
someone has donated $50$ tokens (with a goal of $51$):
\[
\walu{\pmvB}{0}{\tokT} \mid
\walpmv{\contract{Crowdfund}}{\waltok{50}{\tokT},\code{rcv}=\pmvB,\code{goal}=51,\code{isOpen}=\true}
\mid \cdots
\]
For $\pmvB$ to have a positive MEV,
the mempool must contain a transaction where some $\pmvA \neq \pmvB$
donates at least \mbox{$\waltok{1}{\tokT}$}.
For instance, if $\pmvA$ has at least \mbox{$\waltok{10}{\tokT}$} in $\sysS$,
then $\pmvB$ has MEV for
$\TxT = \setenum{\txcode{donate}(\txin{\pmvA}{10}{\tokT})}$.
MEV can be extracted by firing the sequence
$\txcode{donate}(\txin{\pmvA}{10}{\tokT})\,\txcode{claim}() \in \mall{\setenum{\pmvB}}{\emptyset}$, 
resulting in
$\mev{\setenum{\pmvB}}{\sysS}{\TxT} = 60 \cdot \wealth{}{}\idxfun{\tokT} > 0$.

Note however that $\sysS$ is MEV-free for any mempool.
Indeed, choosing a cofinite $\PmvB$ \emph{not} including $\pmvB$
gives $\mev{\PmvA}{\sysSi}{\TxT} = 0$ for any $\PmvA \subseteq \PmvB$
and any $\waldistr{}{\sysS}{}{\sysSi}$,
and so the $\max$ in~\eqref{eq:adv-mev} is zero.
This is coherent with the intuition:
indeed, after the \txcode{init},
the identity of the actor who can extract tokens
is singled out in the contract state,
and it cannot be replaced by an arbitrary miner/validator.
In general, \txcode{Crowdfund} is MEV-free after \txcode{init}
in any finite mempool.
The only case where the contract is not MEV-free
is the (hardly realistic) one
where \txcode{init} has not been performed yet,
and the mempool contains a \txcode{donate}.

%
%

\subsection{Automated Market Makers}
\label{ex:amm}

We now formally prove the well-known fact that
constant-product Automated Market Makers~\cite{Angeris20aft,Angeris21analysis,BCL22amm},
a wide class of decentralized exchanges including mainstream platforms
like Uniswap, Curve and Balancer \cite{uniswapimpl,curveimpl,balancerpaper},
are \emph{not} MEV-free.

\begin{figure}[t]
\begin{lstlisting}[language=txscript,,morekeywords={AMM,addliq,swap0,swap1},classoffset=3,morekeywords={a0,a1,a},keywordstyle=\pmvColor,classoffset=4,morekeywords={T0,T1},keywordstyle=\tokColor,frame=single]
contract AMM {
  addliq(a0 pays x0:T0,a1 pays x1:T1) {
    require balance(T0) * (balance(T1)-x1) == (balance(T0)-x0) * balance(T1)
  }
  swap0(a pays x:T0,ymin) { 
    y = (x * balance(T1)) / balance(T0);
    require y>=ymin && y<balance(T1);
    transfer(a,y:T1); 
  }
  swap1(a pays x:T1,ymin) {
    y = (x * balance(T0)) / balance(T1);
    require y>=ymin && y<balance(T0);
    transfer(a,y:T0);
  }
}
\end{lstlisting}
\negcaptionspace
\caption{A constant-product AMM contract.}
\label{fig:amm}
\end{figure}

Consider the \contract{AMM} contract in~\Cref{fig:amm}.
Users can add reserves of $\tokT[0]$ and $\tokT[1]$
to the contract with \txcode{addliq}
(preserving the reserves ratio), and
exchange units of $\tokT[0]$ with units of $\tokT[1]$
with \txcode{swap0} and \txcode{swap1}.
More specifically,
$\txcode{swap0}(\txin{\pmvA}{x}{\tokT[0]},y_{\min})$
allows $\pmvA$ to send $\waltok{x}{\tokT[0]}$ to the contract,
and receive at least \mbox{$\waltok{y_{\min}}{\tokT[1]}$} in exchange.
Symmetrically, $\txcode{swap1}(\txin{\pmvA}{x}{\tokT[1]},y_{\min})$
allows $\pmvA$ to exchange $\waltok{x}{\tokT[1]}$ for at least
\mbox{$\waltok{y_{\min}}{\tokT[0]}$}.

We show how adversaries can extract MEV from the contract
through the so-called \emph{sandwich attack}~\cite{Zhou21high}.
Assume that the token prices are
$\wealth{}{}\idxfun{\tokT[0]} = 4$
$\wealth{}{}\idxfun{\tokT[1]} = 9$. 
Let:
\[
\sysS =
\walpmv{\contract{AMM}}{\waltok{6}{\tokT[0]},\waltok{6}{\tokT[1]}} \mid
\walu{\pmvA}{3}{\tokT[0]} \mid  
\cdots
\] 
Since the exchange rate given by the AMM
($\waltok{1}{\tokT[0]}$ for $\waltok{1}{\tokT[1]}$)
is more convenient than the exchange rate given by the token prices
($\waltok{9}{\tokT[0]}$ for $\waltok{4}{\tokT[1]}$),
$\pmvA$ would have a positive gain by firing
$\txT[\pmvA] = \txcode{swap1}(\txin{\pmvA}{3}{\tokT[1]},1)$
in the \emph{current} state.
Indeed, we would have:
\begin{align*}
  \sysS
  & \xrightarrow{\txT[\pmvA]} 
  \walpmv{\contract{AMM}}{\waltok{9}{\tokT[0]},\waltok{4}{\tokT[1]}} \mid
  \walpmv{\pmvA}{\waltok{0}{\tokT[0]},\waltok{2}{\tokT[1]}} \mid  
  \cdots
\end{align*}
and so $\pmvA$'s gain would be
\(
\gain{\pmvA}{\sysS}{\txT[\pmvA]}
= 2 \cdot \wealth{}{}\idxfun{\tokT[1]} - 3 \cdot \wealth{}{}\idxfun{\tokT[0]}
= 6 > 0
\).
In a sandwdich attack, the adversary has access to the mempool
$\TxT = \setenum{\txT[\pmvA]}$,
and can have a positive gain to $\pmvA$'s detriment.
Let $\PmvB$ be any cofinite adversary,
and pick $\pmvM \neq \pmvA$ in $\PmvB$.
Assume a token redistribution which assigns $\waltok{3}{\tokT[0]}$ to $\pmvM$,
and let $\PmvA = \setenum{\pmvM}$.
We show that $\PmvA$ has a positive MEV,
and so $(\sysS,\TxT)$ is \emph{not} MEV-free.

The idea of the sandwich attack is the following:
\begin{enumerate}
  
\item $\pmvM$ front-runs $\txT[\pmvA]$ to make the AMM reach
  an equilibrium, where the AMM exchange rate equals
  the exchange rate given by the external token prices.
  This is done through 
  $\txT[\pmvM] = \txcode{swap0}(\txin{\pmvM}{3}{\tokT[0]},2)$;
  
\item then, $\pmvM$ plays $\txT[\pmvA]$.
  In this way, $\pmvA$ will receive fewer units of $\tokT[1]$
  than she would have obtained in $\sysS$.
  Indeed, since the AMM is in equilibrium after $\txT[\pmvM]$,
  $\pmvA$ will have a negative gain;
  
\item finally, $\pmvM$ closes the sandwich with a transaction
  that makes the AMM reach again the equilibrium state.
  This is done through
  $\txTi[\pmvM] = \txcode{swap1}(\txin{\pmvM}{1}{\tokT[1]},3)$.    
\end{enumerate}

\noindent
More precisely, we have the following computation:
\begin{align*}
  \sysS
  & \xrightarrow{\txT[\pmvM]} 
  \walpmv{\contract{AMM}}{\waltok{9}{\tokT[0]},\waltok{4}{\tokT[1]}} \mid
  \walpmv{\pmvM}{\waltok{0}{\tokT[0]},\waltok{2}{\tokT[1]}} \mid
  \walpmv{\pmvA}{\waltok{3}{\tokT[0]}} \mid 
  \cdots
  \\
  & \xrightarrow{\txT[\pmvA]} 
  \walpmv{\contract{AMM}}{\waltok{12}{\tokT[0]},\waltok{3}{\tokT[1]}} \mid
  \walpmv{\pmvM}{\waltok{0}{\tokT[0]},\waltok{2}{\tokT[1]}} \mid
  \walpmv{\pmvA}{\waltok{0}{\tokT[0]},\waltok{1}{\tokT[1]}} \mid 
  \cdots
  \\
  & \xrightarrow{\txTi[\pmvM]} 
  \walpmv{\contract{AMM}}{\waltok{9}{\tokT[0]},\waltok{4}{\tokT[1]}} \mid
  \walpmv{\pmvM}{\waltok{3}{\tokT[0]},\waltok{1}{\tokT[1]}} \mid
  \walpmv{\pmvA}{\waltok{0}{\tokT[0]},\waltok{1}{\tokT[1]}} \mid 
  \cdots  
\end{align*}
The resulting gains for $\pmvM$ and $\pmvA$ are:
\begin{align*}
  \gain{\pmvA}{\sysS}{\txT[\pmvM]\txT[\pmvA]\txTi[\pmvM]}
  & = 1 \cdot \wealth{}{}\idxfun{\tokT[1]} - 3 \cdot \wealth{}{}\idxfun{\tokT[0]}
  = -3 < 0
  \\
  \gain{\pmvM}{\sysS}{\txT[\pmvM]\txT[\pmvA]\txTi[\pmvM]}
  & = 1 \cdot \wealth{}{}\idxfun{\tokT[1]}
  = 9 > 0
\end{align*}
Since this is the optimal strategy for $\pmvM$
(as shown \eg, in~\cite{BCL22fc}),
we conclude that $\mev{}{\sysS}{\TxT} = 9$.

\subsection{Bounty contract}
\label{ex:bounty}

\begin{figure}[t]
  \begin{lstlisting}[language=txscript,morekeywords={BadBounty,init,commit,ver,claim},classoffset=3,morekeywords={a,b,o,A,Oracle},keywordstyle=\pmvColor,classoffset=4,morekeywords={t,T},keywordstyle=\tokColor,frame=single]
contract BadBounty {
  init(b pays n:T) { }
  claim(a,sent_sol) {
    require balance(T)>0;
    require <sent_sol is really the solution>;
    transfer(a,balance(T):T);
  }
}  
  \end{lstlisting}
  \negcaptionspace  
  \caption{A MEV-leaking bounty contract.}
  \label{fig:badbounty}
\end{figure}

\begin{figure}
  \begin{lstlisting}[language=txscript,morekeywords={Bounty,init,commit,ver,claim},classoffset=3,morekeywords={a,b,o,A,Oracle},keywordstyle=\pmvColor,classoffset=4,morekeywords={t,T},keywordstyle=\tokColor,frame=single]
contract Bounty {
  init(b pays n:T) { }
  commit(a,sol_cmt) {
    require !found && balance(T)>0;
    usr[next] = a;
    cmt[next] = sol_cmt;
    next = next+1;
  }
  ver(sent_sol) {
    require !found && balance(T)>0;
    require <sent_sol is really the solution>;
    found = true;     
    sol = sent_sol;
  }
  claim() {
    require found && i_claimed<next;
    if (cmt[i_claimed] == H(usr[i_claimed],sol))
    then transfer(usr[i_claimed],balance(T):T) 
    else i_claimed = i_claimed+1;
  }
}  
  \end{lstlisting}
  \negcaptionspace  
  \caption{A MEV-free bounty contract.}
  \label{fig:bounty}
\end{figure}

Consider a bounty contract which rewards the first user who
submits the solution to a puzzle.
For simplicity we assume that the solution is unique,
and that only a certain actor $\pmvA$ can find the solution.
We start with a na\"ive contract \txcode{BadBounty} in~\Cref{fig:badbounty}.
Here, $\pmvA$ submits the solution by sending
$\txT[\pmvA] = \txcode{claim}(\pmvA,s)$,
with her name and the solution $s$ as parameters.
In this state, where $\txT[\pmvA]$ is in the mempool,
the contract is \emph{not} MEV-free.
Indeed, an adversary $\pmvM$ can front-run $\txT[\pmvA]$
with $\txcode{claim}(\pmvM,s)$, which will release the bounty.

Fixing the contract to make it MEV-free requires some ingenuity.
In the \txcode{Bounty} contract in~\Cref{fig:bounty},
users follow a commit-reveal protocol.
First, they \txcode{commit} their name $\pmvA$ together with the
hash of the pair $(\pmvA,s)$, where $s$ is the bounty solution.
Note that an adversary $\pmvM$ can front-run the commit
replaying the sniffed hashes together with her name $\pmvM$,
but as we will see, this is not enough to extract MEV.
All the commits are recorded in two maps \code{usr} and \code{cmt}:
their ordering is controlled by the adversary.
After $\pmvA$'s commit is finalised on the blockchain,
$\pmvA$ calls \txcode{ver} to submit the actual solution $s$,
making it public.
When $\pmvM$ sees this transaction in the mempool, she can
front-run it with the commit of the pair $(\pmvM,s)$,
which is however recorded \emph{after} $\pmvA$'s commit in the maps.
At this point $\pmvA$ repeatedly sends \txcode{claim} transactions
until receiving the bounty.
The \txcode{claim} procedure scans the maps from the first commit onwards,
releasing the tokens to the \emph{first} user $\pmvB$ who has submitted
the hash of pair $(\pmvB,s)$.
So, even though $\pmvM$ has front-run $\pmvA$'s commit with her own,
she will not receive the bounty, since the hashed name does not correspond
to the name in the \code{usr} map.

In any state $\sysS$ reached when $\pmvA$ follows this protocol,
and any mempool $\TxT$ which contains the next move of $\pmvA$,
the contract is MEV-free.
Indeed, $\mev{\PmvA}{\sysS}{\TxT} = 0$ whenever $\pmvA \not\in \PmvA$,
and so the $\max{}$ in~\eqref{eq:adv-mev} is zero.

\subsection{Lending Pools}
\label{ex:lp}

We analyse MEV in a Lending Pool contract
inspired by Aave~\cite{aave} (see~\Cref{fig:lp}).
In general, Lending Pools implement loan markets
and feature complex economic mechanisms to incentivize users
to deposit tokens and repay loans.
To keep our presentation self-contained, our \txcode{LP} contract
makes several simplifications \wrt mainstream implementations;
these simplifications, however, are irrelevant to the analysis of MEV.

Users can \txcode{deposit} tokens in the \txcode{LP},
obtaining in return virtual tokens minted by the contract.
More precisely, upon a deposit of $\waltok{x}{\tokT}$
in a pool with reserves of $\waltok{n}{\tokT}$,
the user will receive $x \cdot \code{X}(n)$ minted tokens,
where the exchange rate $\code{X}(n)$ is given by:
\[
\code{X}(n) = \begin{cases}
  1 & \text{if $\code{totM}=0$} \\
  \frac{n \, + \, \code{totD}\cdot\code{ir}}{\code{totM}} & \text{otherwise}
\end{cases}
\]
where
\code{totM} is the total number of minted tokens, and
$\code{totD}\cdot\code{ir}$ is the total amount of debt
(more on this below).
A first simplification here is that our \txcode{LP} manages
a \emph{single} token type $\tokT$, while actual implementations
allow \eg, a user to lend tokens of a certain type
and borrow tokens of another type.
Although this simplification makes our \txcode{LP} not interesting 
for practical purposes, as said before it is immaterial for MEV.
The use of minted tokens is twofold.
On the one hand, they are an incentive to lend:
users deposit speculating that the minted tokens will be redeemable 
for a value greater than the original deposit.
On the other hand, they are used as a collateral when borrowing tokens:
namely, users can obtain a loan only if they have enough
\emph{collateralization}, that is given by:
\[
\code{C}(\pmvA,n) = \frac{\code{minted[\pmvA]} \cdot \code{X}(n)}{\code{debt}[\pmvA] \cdot \code{ir}}
\]
where $n$ is the reserve of $\tokT$ in the pool, 
$\code{minted[\pmvA]}$ is the amount of minted tokens owned by the user,
and
$\code{debt[\pmvA]}$ is the amount of $\pmvA$'s debt.
The \txcode{borrow} action requires that the user collateralization
after the action is above a given threshold.

Users can \txcode{redeem} their minted tokens for units of $\tokT$,
where the actual amount is obtained by applying the exchange rate.
Also the redeem requires that the collateralization is
above the threshold.

Interests on loans accrue over time: the current interest rate is given by the exponentiation of a base multiplier \code{Imul}$>1$ to the number of blocks that have been produced since the deployment of the contract.
Note that interest accrual may make some borrowers undercollateralized,
exposing them to liquidations.
Namely, a \txcode{liquidate} action allows \emph{anyone} to repay
part of the debt of an undercollateralized user,
and obtain as a reward part of their minted tokens.
The multiplication factor \code{Rliq}$>1$ incentivizes liquidations. 


\begin{figure}
  \begin{lstlisting}[language=txscript,morekeywords={LP,init,deposit,borrow,accrue,repay,redeem,liquidate},classoffset=3,morekeywords={a,b,o,A,Oracle},keywordstyle=\pmvColor,classoffset=4,morekeywords={t,T},keywordstyle=\tokColor,frame=single]
contract LP {
  fun ir() { // interest rate, which automatically increases over time
    return Imul^(block.number - start)
  }
  fun X(n) { // exchange rate 1:T = X() minted T
    if totM==0 then return 1
    else return (n + totD*ir())/totM
  }
  fun C(a,n) { // collateralization of a
    return (minted[a]*X(n))/(debt[a]*ir())
  }
  init(a sig,c,r,m) { 
    require balance(T)==0 && r>1 && m>1;
    Cmin = c;   // minimum collateralization
    Rliq = r;   // liquidation bonus
    start = block.number;   // starting time
    Imul = m;   // interest rate multiplier
    totD = 0;   // total debt
    totM = 0;   // total minted tokens    
  }
  deposit(a pays x:T) { 
    y = x/X(balance(T)-x); // received minted tokens
    minted[a] += y;
    totM += y;
  } 
  borrow(a sig,x) { 
    require balance(T)>x;
    transfer(a,x:T); 
    debt[a] += x / ir();
    totD += x / ir(); 
    require C(a,balance(T))>=Cmin;
  } 
  repay(a pays x:T) {
    require debt[a]*ir>=x;
    debt[a] -= x / ir(); 
    totD -= x / ir();
  }
  redeem(a sig,x) { 
    y = x * X(balance(T));
    require minted[a]>=x && balance(T)>=y;
    transfer(a,y:T); 
    minted[a] -= x; 
    totM -= x;
    require C(a,balance(T))>=Cmin
  }
  liquidate(a pays x:T,b) { 
    y = (x / X(balance(T)-x)) * Rliq;
    require debt[b]*ir()>x;    
    require C(b,balance(T)-x)<Cmin;
    require minted[b] >= y;
    minted[a] += y; 
    minted[b] -= y; 
    debt[b] -= x / ir();
    totD -= x / ir();
    require C(b,balance(T))<=Cmin;
  }  
}
  \end{lstlisting}
  \negcaptionspace  
  \caption{A Lending Pool contract.}
  \label{fig:lp}
\end{figure}

We now show that it possible to reach a state $S$ where the adversary can extract MEV.
For simplicity, in our example we use fractional values, rather than integers,
and we assume that $\wealth{}{}\idxfun{\tokT} = 1$.
Consider the sequence of transactions:
\begin{itemize}
\item $\txT[0] = \txcode{init}(\txin{\pmv{O}}{0}{\tokT},1.5,1.3,1.5)$,
  that initializes the contract;
\item $\txT[1] = \txcode{deposit}(\txin{\pmvB}{50}{\tokT})$,
  where $\pmvB$ deposits $\waltok{50}{\tokT}$;  
\item $\txT[2] = \txcode{borrow}(\txin{\pmvB}{0}{\tokT},30)$
  where $\pmvB$ borrows $\waltok{30}{\tokT}$.
\end{itemize}

\noindent
From the initial state:
\[
\sysS[0]
\; = \;\;
\walpmv{\contract{LP}}{\waltok{0}{\tokT}} \mid
\walpmv{\pmvM}{\waltok{100}{\tokT}} \mid
\walpmv{\pmvB}{\waltok{100}{\tokT}} \mid 
\cdots
\]
firing $\txT[0]$ leads to the state:
\bartnote{adjust notation}
\begin{align*}
  \sysS[1]
  & =
  \walpmv{\contract{LP}}{\waltok{0}{\tokT},
    \bind{\code{Cmin}}{1.5},\bind{\code{Rliq}}{1.3},\bind{\code{Imul}}{1.5},\cdots} \mid
  \cdots
\end{align*}
where we summarize with $\cdots$ the parts of the state unaffected by the
transition.
Firing $\txT[1]$ in $\sysS[1]$ leads to the state:
\begin{align*}
  \sysS[2]
  & =
  \walpmv{\contract{LP}}{\waltok{50}{\tokT},
    \bind{\code{totM}}{50},
    \bind{\code{minted}}{\pmvB \rightarrow 50},\cdots
  } \mid
  \walpmv{\pmvB}{\waltok{50}{\tokT}} \mid
  \cdots
\end{align*}
The map \code{minted} records that $\pmvB$ now has 50 units of the minted token.
At this point, $\pmvB$ fires $\txT[2]$,
borrowing $\waltok{30}{\tokT}$ from the \txcode{LP}.
In the reached state $\sysS[3]$, these units are transferred from
the contract to $\pmvB$'s wallet,
and the map \code{debt} records that $\pmvB$ has a debt of 30 units:
\begin{align*}
  \sysS[3] 
  & =
  \walpmv{\contract{LP}}{\waltok{20}{\tokT},
    \bind{\code{totM}}{50},
    \bind{\code{totD}}{30},    
    \bind{\code{debt}}{\pmvB \rightarrow 30},\cdots
  } \mid
  \walpmv{\pmvB}{\waltok{80}{\tokT}} \mid 
  \cdots  
\end{align*}

\noindent
Note that the \txcode{borrow} transaction is valid because $\pmvB$'s
collateralization in $\sysS[3]$ is still above the threshold \code{Cmin}:
indeed, we have
$\code{X}(20) = (20+30)/50 = 1$,
and $\pmvB$'s collateralization is
$\code{C}(\pmvB,20) = 50/30 = 1.67 > \code{Cmin} = 1.5$.

Suppose now that the block containing $\txT[0], \txT[1], \txT[2]$ is appended, implicitly updating the interest rate to $\code{ir} = \code{IMul}^1 = 1.5$.
In the blockchain state $\sysS$ at this point, 
$\pmvB$ becomes undercollateralized, since
$\code{X}(20) = (20+(30 \cdot 1.5))/50 = 1.3$,
and $\pmvB$'s collateralization is
$\code{C}(\pmvB,20) = (50 \cdot 1.3)/(30 \cdot 1.5) = 1.444 < \code{Cmin} = 1.5$.

At this point, anyone can liquidate part of $\pmvB$'s debt,
and obtain $\pmvB$'s minted tokens.
To perform the attack, any adversary $\pmvM$ with sufficient balance of $\tokT$ can perform the following sequence of transactions:
\begin{itemize}
\item $\txY[1] = \txcode{liquidate}(\txin{\pmvM}{10}{\tokT},\pmvB)$,
  where $\pmvM$ obtains minted tokens from $\pmvB$'s collateral,
  at a discounted price;
\item $\txY[2] = \txcode{redeem}(\txin{\pmvM}{0}{\tokT},10)$,
  where $\pmvM$ redeems $10$ units of her minted tokens.
\end{itemize}

\noindent
The sequence $\txY[1]\txY[2]$ allows $\pmvM$ to extract MEV from the contract.
The liquidation $\txY[1]$ transfers \mbox{$\waltok{10}{\tokT}$}
from $\pmvM$ to the \txcode{LP}, and updates the maps \code{debt} and \code{minted}:
\begin{align*}
  \sysS[4]
  & =
  \walpmv{\contract{LP}}{\waltok{30}{\tokT},
    \bind{\code{totD}}{23.333},
    \bind{\code{debt}}{\pmvB \rightarrow 23.333},
    \bind{\code{minted}}{\pmvM \rightarrow 10, \pmvB \rightarrow 40},\cdots  
  } \mid
  \walpmv{\pmvM}{\waltok{90}{\tokT}} \mid 
  \cdots  
\end{align*}
Note that after the liquidation, $\pmvB$ is still under-collateralized:
\[
\code{C}(\pmvB,30) = \frac{40 \cdot 1.3}{23.333 \cdot 1.5} = \frac{52}{35} = 1.485 < \code{Cmin} = 1.5
\]
Finally, the transaction $\txY[2]$ allows $\pmvM$ to redeem $10$ units of the minted token with $\code{y} = 10 \cdot \code{X}(30) = 10 \cdot 1.3 = 13$ units of $\tokT$:
\begin{align*}
  \sysS[5]
  & =
  \walpmv{\contract{LP}}{\waltok{17}{\tokT},
    \bind{\code{minted}}{\pmvM \rightarrow 0, \pmvB \rightarrow 40},\cdots  
  } \mid
  \walpmv{\pmvM}{\waltok{103}{\tokT}} \mid 
  \cdots  
\end{align*}
Summing up, $\pmvM$'s gain is $(103 - 100) \cdot \wealth{}{}\idxfun{\tokT} = 3 \cdot \wealth{}{}\idxfun{\tokT}$.
Note that the identity of the adversary is immaterial for the attack.
Therefore, $\mev{}{\sysS}{\emptyset} =  3 \cdot \wealth{}{}\idxfun{\tokT}$.

\section{Supplementary results and proofs}
\label{sec:proofs}

%
%

\subsection{Blockchain model (\Cref{sec:contracts})}

We detail below the proofs of the statements in~\Cref{sec:contracts},
together with additional results.

The finite tokens axiom ensures that the wallet of any
set of actors $\PmvA$, defined as the pointwise addition of functions:
\[
\wal{\PmvA}{\sysS}
\; = \;
\textstyle
\sum_{\pmvA \in \PmvA}{\wal{\pmvA}{\sysS}}
\]
is always defined,
and that it has a non-zero amount of tokens only for a \emph{finite}
set of token types.

\begin{prop}
  \label{prop:wal:PmvA}
  For all (possibly infinite) $\PmvA$ and $\sysS$,
  $\wal{\PmvA}{\sysS}$ is defined and has a finite support.
\end{prop}
\begin{proof} 
  Direct consequence of~\eqref{eq:finite-tokens}.
\end{proof}


Furthermore, for any possibly infinite $\PmvA$,
there exists a least \emph{finite} subset $\PmvAfin$ which contains exactly
the tokens of $\PmvA$.
This is also true for all subsets $\PmvB$ of $\PmvA$ including $\PmvAfin$.

\begin{prop}
  \label{prop:wal:finite-actors}
  For all $\sysS$, $\PmvA$,
  there exists the least $\PmvAfin \subseteqfin \PmvA$ such that,
  for all $\PmvB$, if $\PmvAfin \subseteq \PmvB \subseteq \PmvA$
  then $\wal{\PmvAfin}{\sysS} = \wal{\PmvB}{\sysS}$.
\end{prop}
\begin{proof} 
  By definition,
  $\wal{\PmvA}{\sysS} \tokT = \sum_{\pmvA \in \PmvA}{\wal{\pmvA}{\sysS} \tokT}$.
  By~\eqref{eq:finite-tokens}, 
  the quantity $\wal{\pmvA}{\sysS} \tokT$ is non-zero for only finitely many 
  $\pmvA$ and $\tokT$.
  Therefore, the set 
  $\PmvAfin = \setcomp{\pmvA}{\exists \tokT . \ \wal{\pmvA}{\sysS} \tokT \neq 0}$ 
  is finite, and $\wal{\PmvA \setminus \PmvAfin}{\sysS} \tokT = 0$ for all $\tokT$.
  Consequently,
  $\wal{\PmvA}{\sysS} \tokT = \wal{\PmvAfin}{\sysS} \tokT \leq \wal{\PmvB}{\sysS} \tokT \leq \wal{\PmvA}{\sysS} \tokT$.
  Note that, by construction, $\PmvAfin$ is the least subset of $\PmvA$ preserving the tokens.
  \qed
\end{proof}

Hereafter, we denote by $\finwalU$ the set of finite-support functions
from $\TokU$ to $\Nat$, like those resulting from $\wal{\PmvA}{\sysS}$.

We can extend~\eqref{eq:wealth:token-prices:wealth} to
the sum of a countable set of wallets:

\begin{prop}
  \label{prop:wealth:token-prices}
  Let $(\wmvA[i])_i$ be a countable set of wallets in $\finwalU$ such that
  $\sum_{i,\tokT} \wmvA[i](\tokT) \in \Nat$.
  Then:
  \[
  \wealth{}{} \sum_i \wmvA[i]
  \; = \;
  \sum_i \wealth{}{} \wmvA[i]
  \; = \;
  \sum_{i,\tokT} \wmvA[i](\tokT) \cdot \wealth{}{} \idxfun{\tokT}
  \]
\end{prop}
\begin{proof} 
  Since $\sum_i \wmvA[i]$ is in $\TokU \rightarrow \Nat$,
  by~\eqref{eq:wealth:token-prices:wal} we can write:
  \begin{equation}
    \label{prop:wealth:token-prices:eq1}
    \sum_i \wmvA[i]
    \; = \;
    \sum_{\tokT} \Big( \sum_i \wmvA[i] \Big)(\tokT) \cdot \idxfun{\tokT}
    \; = \;
    \sum_{\tokT,i} \wmvA[i](\tokT) \cdot \idxfun{\tokT}
  \end{equation}
  By hypothesis, $\wmvA[i](\tokT)$ is almost always zero,
  so the rhs in~\eqref{prop:wealth:token-prices:eq1} is a \emph{finite} sum.
  Then:
  \begin{align*}
    \wealth{}{} \sum_i \wmvA[i]  
    & = \wealth{}{} \sum_{\tokT,i} \wmvA[i](\tokT) \cdot \idxfun{\tokT}
    && \text{by~\eqref{prop:wealth:token-prices:eq1}}
    \\
    & = \sum_{\tokT,i} \wealth{}{} \big(\wmvA[i](\tokT) \cdot \idxfun{\tokT}\big)
    && \text{additivity, finite sum}
    \\
    & = \sum_{\tokT,i} \wmvA[i](\tokT) \cdot \wealth{}{}\idxfun{\tokT}
    && \text{additivity, finite sum}
    \\
    & = \sum_{i} \sum_{\tokT} \wmvA[i](\tokT) \cdot \wealth{}{}\idxfun{\tokT}
    \\
    & = \sum_{i} \sum_{\tokT} \wealth{}{} \big( \wmvA[i](\tokT) \cdot \idxfun{\tokT} \big)
    && \text{additivity, finite sum}
    \\     
    & = \sum_{i} \wealth{}{} \sum_{\tokT} \wmvA[i](\tokT) \cdot \idxfun{\tokT}
    && \text{additivity, finite sum}
    \\     
    & = \sum_{i} \wealth{}{} \wmvA[i]
    && \text{by~\eqref{eq:wealth:token-prices:wal}}
    \tag*{\qedhere}
  \end{align*}
\end{proof}

Note that, since $\PmvU$ is countable, so is any subset $\PmvA$ of $\PmvU$.
Therefore, we have the following corollary of~\Cref{prop:wealth:token-prices}.

\begin{prop}
  \label{cor:wealth:finite-sum}
  For all $\sysS$ and $\PmvA$,
  $\wealth{\PmvA}{\sysS} = \sum_{\pmvA \in \PmvA} \wealth{\pmvA}{\sysS}$.
\end{prop}

We can extend \Cref{prop:wal:finite-actors} to wealth,
by ensuring that, in any reachable state,
only \emph{finitely} many actors have a non-zero wealth:

\begin{prop}
  \label{prop:wealth:finite-actors}
  For all $\sysS$, $\PmvA$,
  there exists the least $\PmvAfin \!\subseteqfin \!\PmvA$ such that,
  for all $\PmvB$, if $\PmvAfin \subseteq \! \PmvB \subseteq\! \PmvA$
  then $\wealth{\PmvAfin}{\sysS} = \wealth{\PmvB}{\sysS}$.
\end{prop}
\begin{proof} 
  Let $\PmvA[1] = \setenum{\pmvA[1],\ldots,\pmvA[n]}$
  be the finite set given by~\Cref{prop:wal:finite-actors},
  for which we have:
  \(
  \wal{\PmvA}{\sysS} = \wal{\PmvA[1]}{\sysS} = \sum_{i = 1}^n \wal{\pmvA[i]}{\sysS}
  \).
  Since this is a finite sum, by additivity of $\wealth{}{}$ we obtain:
  \(
  \wealth{\PmvA}{\sysS} = \wealth{\PmvA[1]}{\sysS} = \sum_{i = 1}^n \wealth{\pmvA[i]}{\sysS}
  \).
  Now, let $\PmvAfin = \setcomp{\pmvA \in \PmvA[1]}{\wealth{\pmvA}{\sysS} \neq 0}$.
  Since $\PmvAfin$ is finite, we have
  $\wealth{\PmvA}{\sysS} = \wealth{\PmvA[1]}{\sysS} = \wealth{\PmvAfin}{\sysS}$.
  Now, let $\PmvAfin \subseteq \PmvB \subseteq \PmvA$.
  Since $\wealth{}{}$ is additive, it is also monotonic, and so
  $\wealth{\PmvA}{\sysS} = \wealth{\PmvAfin}{\sysS} \leq \wealth{\PmvB}{\sysS} \leq \wealth{\PmvA}{\sysS}$.
  Note that, by construction, $\PmvAfin$ is the least set satisfying the statement.
\end{proof}




    
    

\begin{proofof}{prop:gain}
  We must prove the following statements:
  \begin{enumerate}[\rm (1)]

  \item \label{prop:gain:PmvA}
    $\gain{\PmvA}{\sysS}{\TxTS}$ is defined and has a finite integer value

  \item \label{prop:gain:sum}
    $\gain{\PmvA}{\sysS}{\TxTS} = \sum_{\pmvA \in \PmvA} \gain{\pmvA}{\sysS}{\TxTS}$
    
    \item \label{prop:gain:fin}
      there exists $\PmvAfin \subseteqfin\! \PmvA$ such that,
      for all $\PmvB$, if $\PmvA[0] \subseteq \PmvB \subseteq \PmvA$ then
      $\gain{\PmvAfin}{\sysS}{\TxTS} = \gain{\PmvB}{\sysS}{\TxTS}$.
    
  \end{enumerate}

  \Cref{prop:gain:PmvA} follows from~\Cref{prop:wal:PmvA}.

  \noindent
  \Cref{prop:gain:sum} follows from Propositions~\ref{prop:wal:PmvA} and~\ref{cor:wealth:finite-sum}.

  \noindent
  For~\Cref{prop:gain:fin},
  let $\sysS \xrightarrow{\TxTS} \sysSi$.
  By~\Cref{def:gain}, we have that
  $\gain{\PmvA}{\sysS}{\TxTS} = \wealth{\PmvA}{\sysSi} - \wealth{\PmvA}{\sysS}$.
  By~\Cref{prop:wealth:finite-actors},
  there exist $\PmvA[1],\PmvA[2] \subseteqfin \PmvA$ such that:  
  \begin{align}
    \label{prop:gain:fin:1}
    & \forall \PmvB : \PmvA[1] \subseteq \PmvB \subseteq \PmvA \implies
    \wealth{\PmvA[1]}{\sysS} = \wealth{\PmvB}{\sysS}
    \\
    \label{prop:gain:fin:2}
    & \forall \PmvB : \PmvA[2] \subseteq \PmvB \subseteq \PmvA \implies    
    \wealth{\PmvA[2]}{\sysSi} = \wealth{\PmvB}{\sysSi}
  \end{align} 
  Let $\PmvA[0] = \PmvA[1] \cup \PmvA[2]$,
  and let $\PmvB$ be such that $\PmvA[0] \subseteq \PmvB \subseteq \PmvA$.
  Therefore:
  \begin{align*}
    \gain{\PmvB}{\sysS}{\TxTS}
    & = \wealth{\PmvB}{\sysSi} - \wealth{\PmvB}{\sysS}
    && \text{by~\Cref{def:gain}}
    \\
    & = \wealth{\PmvA[0]}{\sysSi} - \wealth{\PmvA[0]}{\sysS}
    && \text{by~\eqref{prop:gain:fin:1} and~\eqref{prop:gain:fin:2}}
    \\
    & = \gain{\PmvA[0]}{\sysS}{\TxTS}
    \tag*{\qedhere}
  \end{align*}
\end{proofof}

%
%

\subsection{Adversarial knowledge (\Cref{sec:tx-inference})}

\begin{proofof}{lem:mall:basic}
  For~\ref{lem:mall:inj}, 
  if $\mall{\PmvA}{\emptyset} = \mall{\PmvB}{\emptyset}$,
  then by the private knowledge axiom
  $\PmvA \subseteq \PmvB$ and $\PmvB \subseteq \PmvA$, and so $\PmvA = \PmvB$.  
  Item~\ref{lem:mall:cap} follows from no shared secrets and monotonicity.
  The leftmost inclusion of item~\ref{lem:mall:cup} follows from extensivity and monotonicity,
  while the rightmost inclusion is given by:
  \begin{align*}
    \mall{\PmvA}{\mall{\PmvB}{\TxT}}
    & \subseteq 
      \mall{\PmvA \cup \PmvB}{\mall{\PmvA \cup \PmvB}{\TxT}}
    && \text{monotonicity, twice}
    \\
    & \subseteq 
      \mall{\PmvA \cup \PmvB}{\TxT}
    && \text{idempotence}
  \end{align*}
  For~\ref{lem:mall:subseteq-mall-mall}, we have:
  \begin{align*}
    \mall{\PmvB}{\mall{\PmvA}{\TxY} \cup \TxT}
    & \subseteq
    \mall{\PmvB}{\mall{\PmvA}{\TxT} \cup \TxT}
    && \text{monotonicity, $\TxY \subseteq \TxT$}
    \\
    & =
    \mall{\PmvB}{\mall{\PmvA}{\TxT}}
    && \text{extensivity}
    \\
    & \subseteq
    \mall{\PmvA}{\mall{\PmvA}{\TxT}}
    && \text{monotonicity, $\PmvB \subseteq \PmvA$}
    \\
    & = \mall{\PmvA}{\TxT}
    && \text{idempotence} 
  \end{align*}

  For~\Cref{lem:mall:finite-tx},
  let $\TxT = \bigcup_i \TxT[i]$, where $\setenum{\TxT[i]}_i$ is an increasing chain
  of finite sets.
  By continuity, we have that:
  \[
  \txT \in \mall{\PmvA}{\TxT}
  \textstyle
  = \mall{\PmvA}{\bigcup_i \TxT[i]}
  = \bigcup_i \mall{\PmvA}{\TxT[i]}
  \]
  Therefore, there exists some $n$ such that $\txT \in \mall{\PmvA}{\TxT[n]}$.
  The thesis follows from the fact that $\TxT[n]$ is a finite subset of $\TxT$.
  \qed
\end{proofof}

\begin{proofof}{lem:auth:basic}
  Item~\ref{lem:auth:empty} is straightforward by~\Cref{def:auth}.
  
  Item~\ref{lem:auth:mall-auth} is a direct consequence of extensivity.

  For~\ref{lem:auth:from-empty-tx},
  \wwlog, take a minimal finite set of actors $\PmvAfin$
  such that $\TxT \subseteq \mall{\PmvAfin}{\emptyset}$
  (this is always possible, since finite sets are well ordered by inclusion).
  We prove that
  $\PmvAfin = \min \setcomp{\PmvB}{\TxT \subseteq \mall{\PmvB}{\emptyset}}$.
  By contradiction, let $\PmvA$ be such that
  $\TxT \subseteq \mall{\PmvA}{\emptyset}$ and
  $\PmvAfin \nsubseteq \PmvA$.
  Let $\PmvB = \PmvAfin \cap \PmvA$.
  We have that:
  \begin{align*}
    \TxT
    & \subseteq
    \mall{\PmvAfin}{\emptyset} \cap \mall{\PmvA}{\emptyset}
    && \text{$\TxT \subseteq \mall{\PmvAfin}{\emptyset}$, $\TxT \subseteq \mall{\PmvA}{\emptyset}$}
    \\
    & \; \subseteq \;
    \mall{\PmvAfin \cap \PmvA}{\emptyset}
    && \text{no shared secrets}
    \\
    & \; = \;
    \mall{\PmvB}{\emptyset}
    && \text{$\PmvB = \PmvAfin \cap \PmvA$}
  \end{align*}
   Since $\PmvAfin \nsubseteq \PmvA$, then $\PmvB \subset \PmvAfin$,
  and so $\PmvAfin$ is not minimal --- contradiction.

  For~\ref{lem:auth:from-finite-part},
  by the finite causes property, there exists a finite $\PmvAfin$ such that
  $\TxT \subseteq \mall{\PmvAfin}{\emptyset}$.
  The thesis follows by~\Cref{lem:auth:from-empty-tx}.
  \qed
\end{proofof}

%
%

\subsection{MEV (\Cref{sec:mev})}

\begin{proofof}{prop:mev:defined}
  By~\Cref{def:mall}, $\mall{\PmvA}{\TxT}$ is defined.
  By \mbox{$\wealth{}{}$-boundedness}, $\gain{\PmvA}{\sysS}{\TxTS}$ 
  is defined and bounded for all $\TxTS$.
  The existence of the maximum in~\Cref{def:mev} follows 
  from the fact that each non-empty upper-bounded subset of $\Nat$ admits a maximum.
  This maximum exists even though $\mall{\pmvA}{-}^*$ generates
  an infinite set of sequences, of unbounded (but finite) length.
  The MEV is non-negative since the empty sequence of transactions
  gives a zero gain.
  \qed
\end{proofof}

We denote by $\auth{\TxT}$ the least set of actors who are needed to  
craft all the transactions in $\TxT$. 
When $\auth{\TxT} = \emptyset$, then $\TxT$ can be created by anyone,
without requiring any private knowledge.

\begin{defn}[Transaction authoriser]
  \label{def:auth}
  For all $\TxT$, we define:
  \[
  \auth{\TxT}
  \; = \;
  \min \setcomp{\PmvA}{\TxT \subseteq \mall{\PmvA}{\emptyset}}
  \]
  when such a minimum exists.
\end{defn}

\Cref{lem:auth:basic} establishes some properties of authorisers.
\Cref{lem:auth:mall-auth} states that 
the authorisers of $\TxT$ can generate at least $\TxT$.
\Cref{lem:auth:from-empty-tx,lem:auth:from-finite-part} give sufficient conditions
for $\auth{\TxT}$ to be defined:
indeed, $\auth{\TxT}$ may be undefined for some $\TxT$,
since the set of $\PmvA$ such that $\TxT \subseteq \mall{\PmvA}{\emptyset}$
may not have a minimum.
\Cref{lem:auth:from-empty-tx}
guarantees that $\auth{\TxT}$ is defined (and finite)
whenever $\TxT$ is inferred from an empty knowledge.
\Cref{lem:auth:from-finite-part}
gives the same guarantee when $\TxT$ is finite.

\begin{prop}
  \label{lem:auth:basic}
  For all $\TxT$ and $\PmvA$, we have that:
  \begin{enumerate}[\rm(1)]

  \item \label{lem:auth:empty}
    $\auth{\emptyset} = \emptyset$
    
  \item \label{lem:auth:mall-auth}
    if $\auth{\TxT}$ is defined, then $\TxT \subseteq \mall{\auth{\TxT}}{\emptyset}$

  \item \label{lem:auth:from-empty-tx}
    if $\TxT \subseteq \mall{\PmvA}{\emptyset}$ with $\PmvA$ finite, 
    $\auth{\TxT}$ is defined and~finite
    
  \item \label{lem:auth:from-finite-part}
    if $\TxT$ is finite, then $\auth{\TxT}$ is defined and finite.
  
\end{enumerate}
\end{prop}

For a \txscript transaction $\txT = \procF(\overrightarrow{\textit{txarg}})$,
the authorizers are the $\pmvA$ such that
$\overrightarrow{\textit{txarg}}$ contains
$\txin{\pmvA}{n}{\tokT}$, for some $n \in \Nat$ and $\tokT \in \TokU$,
or involves a secret $s$ generated by $\pmvA$.
For instance,
in \txcode{BadHTLC} (\Cref{ex:txscript:badhtlc:mall}) we have:
$\auth{\txT[\pmvA]} = \setenum{\pmvA}$,
$\auth{\txT[\pmvM]} = \setenum{\pmvM,\pmvA}$, and
$\auth{\txY[\pmvM]} = \setenum{\pmvM,\pmv{Oracle}}$.


A special case of~\Cref{prop:mev:cut} is when
the actors in $\PmvA$ have enough knowledge to generate the
set $\TxT$ by themselves.
In this case, the $\mev{}{}{}$ extractable by $\PmvA$ is 
that of a displacement attack:
\begin{cor}
  \label{prop:mev:zero}
  If $\auth{\TxT} \subseteq \PmvA$, then $\mev{\PmvA}{\sysS}{\TxT} = \mev{\PmvA}{\sysS}{\emptyset}$.
\end{cor}
\begin{proof}
  Consequence of~\Cref{prop:mev:cut}, using the inequality:
  \begin{align*}
    \mall{\PmvA}{\TxT}
    & \subseteq 
    \mall{\PmvA}{\mall{\auth{\TxT}}{\emptyset}}
    && \text{\Cref{lem:auth:basic}\ref{lem:auth:mall-auth}}
    \\
    & \subseteq 
    \mall{\PmvA \cup \auth{\TxT}}{\emptyset}
    && \text{\Cref{lem:mall:basic}\ref{lem:mall:mall}}
    \\
    & \subseteq 
    \mall{\PmvA}{\emptyset}	
    && \text{$\auth{\TxT} \subseteq \PmvA$}
    \tag*{\qedhere}
  \end{align*}
\end{proof}

\begin{proofof}{prop:mev:cut}
  We have that $\TxT = \TxY \cup \TxZ$,
  where $\TxY = \TxT \setminus \mall{\PmvA}{\emptyset}$ and
  $\TxZ = \TxT \cap \mall{\PmvA}{\emptyset}$.
  We show that 
  $\mall{\PmvA}{\TxT} = \mall{\PmvA}{\TxY}$.
  The inclusion $\supseteq$ follows by monotonicity of $\mall{}{}$.
  For $\subseteq$, we have that:
  \begin{align*}
    \mall{\PmvA}{\TxT}
    & = \mall{\PmvA}{\TxY \cup \TxZ}
    && \text{$\TxT = \TxY \cup \TxZ$}
    \\
    & \subseteq
    \mall{\PmvA}{\TxY \cup \mall{\PmvA}{\emptyset}}
    && \text{monotonicity, $\TxZ \subseteq \mall{\PmvA}{\emptyset}$}
    \\
    & \subseteq
    \mall{\PmvA}{\TxY \cup \mall{\PmvA}{\TxY}}
    && \text{monotonicity}
    \\
    & =
    \mall{\PmvA}{\mall{\PmvA}{\TxY}}
    && \text{extensivity}
    \\
    & =
    \mall{\PmvA}{\TxY}
    && \text{idempotence}
    \tag*{\qed}
  \end{align*}
\end{proofof}

\begin{proofof}{prop:mev:monotonic-on-tx}
  Let $\TxTS \in \mall{\PmvA}{\TxT}^*$ maximize $\gain{\PmvA}{\sysS}{\cdot}$.
  By monotonicity of $\mall{}{}$ we have that  
  $\mall{\PmvA}{\TxT} \subseteq \mall{\PmvAi}{\TxTi}$.
  Hence, $\TxTS \in \mall{\PmvAi}{\TxTi}^*$.
  The thesis follows from \Cref{prop:gain}\ref{prop:gain:sum}.
  \qed
\end{proofof}

\begin{proofof}{th:mev:finite-part}
  Let $\PmvB \subseteq \PmvA$.
  By~\Cref{def:mev}, we have that:
  \[
  \mev{\PmvB}{\sysS}{\TxT}
  \; = \;
  \gain{\PmvB}{\sysS}{\TxTS}
  \tag*{for some $\TxTS \in \mall{\PmvB}{\TxT}^*$}
  \]
  Since $\TxTS$ is made of a finite set of transactions,
  by the finite causes property of $\mall{}{}$
  there exists a \emph{finite} $\PmvC$ such that
  $\TxTS \in \mall{\PmvC}{\emptyset}^*$.
  Then, by monotonicity of $\mall{}{}$ it follows that
  \(
  \TxTS \in \mall{\PmvC}{\emptyset}^*
  \subseteq \mall{\PmvC}{\TxT}^*
  \).
  Let $\PmvC[0] = \PmvC \cap \PmvB$.  
  By the ``no shared secrets'' property of $\mall{}{}$:
  \begin{align*}
    \TxTS
    & \in \mall{\PmvC}{\TxT}^* \cap \mall{\PmvB}{\TxT}^*
    = (\mall{\PmvC}{\TxT} \cap \mall{\PmvB}{\TxT})^*
    \\
    & \subseteq \mall{\PmvC \cap \PmvB}{\TxT}^*
    = \mall{\PmvC[0]}{\TxT}^*
  \end{align*}

  By~\Cref{prop:gain}\ref{prop:gain:fin},
  there exists a finite $\PmvC[1] \subseteqfin \PmvB$ such that
  $\gain{\PmvB}{\sysS}{\TxTS} = \gain{\PmvC[1]}{\sysS}{\TxTS}$.
  Let $\PmvAfin = \PmvC[0] \cup \PmvC[1] \subseteq \PmvB$.
  We have that:
  \begin{align*}
    \gain{\PmvB}{\sysS}{\TxTS}
    & = \gain{\PmvC[1]}{\sysS}{\TxTS}
    = \gain{\PmvAfin}{\sysS}{\TxTS}
  \end{align*}
  Since $\TxTS \in \mall{\PmvC[0]}{\TxT}^*$,
  the monotonicity of $\mall{}{}$ implies 
  $\TxTS \in \mall{\PmvAfin}{\TxT}^*$.
  Summing up:
  \begin{align*}
    \mev{\PmvB}{\sysS}{\TxT}
    & \; = \; \gain{\PmvB}{\sysS}{\TxTS}
    \; = \; \gain{\PmvAfin}{\sysS}{\TxTS}
    \; = \; \mev{\PmvAfin}{\sysS}{\TxT}
    \tag*{\qed}
  \end{align*}
\end{proofof}

\begin{proofof}{th:mev:finite-tx}
  By~\Cref{def:mev}, we have that:
  \[
  \mev{\PmvA}{\sysS}{\TxT}
  \; = \;
  \gain{\PmvA}{\sysS}{\TxTS}
  \tag*{for some $\TxTS = \txT[1] \cdots \txT[n] \in \mall{\PmvA}{\TxT}^*$}
  \]
  By~\Cref{lem:mall:basic}\ref{lem:mall:finite-tx},
  for all $i \in 1..n$ there exists $\TxT[i] \subseteqfin \TxT$ such that
  $\txT[i] \in \mall{\PmvA}{\TxT[i]}$.
  Let $\TxTfin = \bigcup_{i \in 1..n} \TxT[i] \subseteqfin \TxT$.
  By monotonicity of $\mall{}{}$, we have
  $\txT[i] \in \mall{\PmvA}{\TxTfin}$ for all $i \in 1..n$.
  Therefore,
  $\TxTS \in \mall{\PmvA}{\TxTfin}^*$.
  \qed
\end{proofof}

%
%

\subsection{Bad MEV (\Cref{sec:mev})}

\begin{prop}
  \label{prop:badmev:defined}
  $\badmev{\PmvA}{\sysS}{\TxT}$ is defined and has a non-negative value.
\end{prop}
\begin{proof}
    Immediate from~\Cref{prop:mev:defined}, \Cref{prop:mev:monotonic-on-tx}, and~\Cref{def:badmev}.
    \qed
\end{proof}

\begin{prop}
  \label{prop:badmev:cut}
  \(
  \badmev{\PmvA}{\sysS}{\TxT} = \badmev{\PmvA}{\sysS}{\TxT \setminus \mall{\PmvA}{\emptyset}}
  \)
\end{prop}
\begin{proof}
We have that:
\begin{align*}
  \badmev{\PmvA}{\sysS}{\TxT} 
  & = \mev{\PmvA}{\sysS}{\TxT} - \mev{\PmvA}{\sysS}{\emptyset}
  && \text{by~\Cref{def:badmev}}
  \\
  & = \mev{\PmvA}{\sysS}{\TxT \setminus \mall{\PmvA}{\emptyset}} - \mev{\PmvA}{\sysS}{\emptyset}
  && \text{by~\Cref{prop:mev:cut}}
  \\
  & = \badmev{\PmvA}{\sysS}{\TxT \setminus \mall{\PmvA}{\emptyset}}
  && \text{by~\Cref{def:badmev}}
  \tag*{\qedhere}
\end{align*}
\end{proof}

\begin{prop} 
  \label{prop:badmev:monotonic-on-tx}
  If $\TxT \subseteq \TxTi$, then \
  \(
  \badmev{\PmvA}{\sysS}{\TxT} \leq \badmev{\PmvA}{\sysS}{\TxTi}
  \).
\end{prop}
\begin{proof}
We have that:
\begin{align*}
  \badmev{\PmvA}{\sysS}{\TxT} 
  & = \mev{\PmvA}{\sysS}{\TxT} - \mev{\PmvA}{\sysS}{\emptyset}
  && \text{by~\Cref{def:badmev}}
  \\
  & \leq \mev{\PmvA}{\sysS}{\TxTi} - \mev{\PmvA}{\sysS}{\emptyset}
  && \text{by~\Cref{prop:mev:monotonic-on-tx}}
  \\
  & = \badmev{\PmvA}{\sysS}{\TxTi}
  && \text{by~\Cref{def:badmev}}
  \tag*{\qedhere}
\end{align*}
\end{proof}


\begin{prop} 
  \label{th:badmev:finite-part}
  For all $\PmvA$, $\TxT$ and $\sysS$,
  there exists $\PmvAfin \subseteqfin \PmvA$
  such that, for all $\PmvB$, if $\PmvAfin \subseteq \PmvB \subseteq \PmvA$
  then: \
  \(
  \badmev{\PmvAfin}{\sysS}{\TxT} = \badmev{\PmvB}{\sysS}{\TxT}
  \).
\end{prop}
\begin{proof}
By~\Cref{th:mev:finite-part} (applied twice), there exist
$\PmvA[0] \subseteqfin \PmvA$
and
$\PmvA[1] \subseteqfin \PmvA$
such that, for all $\PmvB[0]$ and $\PmvB[1]$: 
\begin{align}
\label{eq:badmev:finite-part:1}
\PmvA[0] \subseteq \PmvB[0] \subseteq \PmvA    
& \implies
\mev{\PmvA[0]}{\sysS}{\TxT} = \mev{\PmvB[0]}{\sysS}{\TxT}
\\
\label{eq:badmev:finite-part:2}
\PmvA[1] \subseteq \PmvB[1] \subseteq \PmvA    
& \implies
\mev{\PmvA[1]}{\sysS}{\emptyset} = \mev{\PmvB[1]}{\sysS}{\emptyset}
\end{align}
Let $\PmvA[2] = \PmvA[0] \cup \PmvA[1]$, which is finite.
Then, for all $\PmvB$ such that $\PmvA[2] \subseteqfin \PmvB \subseteq \PmvA$:
  \begin{align*}
  \badmev{\PmvA}{\sysS}{\TxT}
  & =
  \mev{\PmvA}{\sysS}{\TxT} - \mev{\PmvA}{\sysS}{\emptyset}
  && \text{by~\Cref{def:badmev}}
  \\
  & = 
  \mev{\PmvB}{\sysS}{\TxT} - \mev{\PmvB}{\sysS}{\emptyset}
  && \text{by~\eqref{eq:badmev:finite-part:1},\eqref{eq:badmev:finite-part:2}}
  \\
  & = 
  \badmev{\PmvB}{\sysS}{\TxT}
  \tag*{\qedhere}
  \end{align*}
\end{proof}

\begin{prop} 
  \label{th:badmev:finite-tx}
  For all $\PmvA$, $\TxT$ and $\sysS$, there exists $\TxTfin \subseteqfin \TxT$ such that
  $\badmev{\PmvA}{\sysS}{\TxTfin} = \badmev{\PmvA}{\sysS}{\TxT}$.
\end{prop}
\begin{proof}
  By~\Cref{th:badmev:finite-tx}, there exists $\TxTfin \subseteqfin \TxT$ such that
  $\mev{\PmvA}{\sysS}{\TxTfin} = \mev{\PmvA}{\sysS}{\TxT}$.
  Therefore, by~\Cref{def:badmev}:
  \begin{align*}
  \badmev{\PmvA}{\sysS}{\TxT}
  & =
  \mev{\PmvA}{\sysS}{\TxT} - \mev{\PmvA}{\sysS}{\emptyset}
  \\
  & = 
  \mev{\PmvA}{\sysS}{\TxTfin} - \mev{\PmvA}{\sysS}{\emptyset}
  \\
  & = 
  \badmev{\PmvA}{\sysS}{\TxTfin}
  \tag*{\qedhere}
  \end{align*}
\end{proof}

\begin{proofof}{prop:badlem}
Consequence of
\Cref{prop:badmev:defined},
\Cref{prop:badmev:cut},
\Cref{prop:badmev:monotonic-on-tx}, 
\Cref{th:badmev:finite-part},
\Cref{th:badmev:finite-tx}.
\qed
\end{proofof}

%
%

\subsection{Preservation under renaming}

In general, renaming actors in $\PmvA$ may affect $\mev{\PmvA}{\sysS}{\TxT}$.
However, in most real-world contracts there exist sets of actors 
that are \emph{indistinguishable} from each other when looking at their
interaction capabilities with the contract.
We formalise as \emph{clusters} these sets of actors, 
and we show in \Cref{th:mev:cluster-renaming} that MEV
is invariant \wrt renaming actors in a cluster. 

\Cref{def:renaming} formalises the effect of renaming
a set of actors $\PmvA$ on states.
We say that $\rho$ is a \emph{renaming of $\PmvA$}
if $\rho \in \PmvA \rightarrow \PmvA$ is a permutation.
Applying $\rho$ to a blockchain state $(\WmvA,\ContrS)$
only affects the wallet state, leaving the contract state unaltered.
The renamed wallet state $\WmvA\rho$ is defined as follows.

\begin{defn}
  \label{def:renaming}
  For a wallet state $\WmvA$ and renaming $\rho$ of $\PmvA$, let:
  \[
  (\WmvA \rho)(\pmvA) = \begin{cases}
    \WmvA (\rho^{-1}(\pmvA)) & \text{if $\pmvA \in \PmvA$} \\
    \WmvA(\pmvA) & \text{otherwise}
  \end{cases}
  \]
  We define the renaming of $(\WmvA,\ContrS)$ as
  $(\WmvA,\ContrS) \rho = (\WmvA \rho,\ContrS)$.
\end{defn}

As a basic example, for
\(
\WmvA =
\walu{\pmvA[0]}{0}{\tokT} \mid
\walu{\pmvA[1]}{1}{\tokT} \mid
\walu{\pmvA[2]}{2}{\tokT}
\)
and
\(
\rho = \setenum{\bind{\pmvA[0]}{\pmvA[1]}, \bind{\pmvA[1]}{\pmvA[2]}, \bind{\pmvA[2]}{\pmvA[0]}}
\)
is a renaming of $\setenum{\pmvA[0],\pmvA[1],\pmvA[2]}$, then
$\WmvA\rho = \walu{\pmvA[1]}{0}{\tokT} \mid \walu{\pmvA[2]}{1}{\tokT} \mid
  \walu{\pmvA[0]}{2}{\tokT}$.

Intuitively, a set $\PmvA$ is a cluster when, for each renaming $\rho$
of $\PmvA$ and each subset $\PmvB$ of $\PmvA$,
if $\PmvB$ can fire a sequence of transactions leading to a given wallet state,
then also $\rho(\PmvB)$ can achieve the same effect, up-to the renaming $\rho$.

\begin{defn}[Cluster of actors]
  \label{def:cluster}
  A set of actors $\PmvA \subseteq \PmvU$ is a
  \emph{cluster in $(\sysS, \TxT)$} 
  iff, for all renamings $\rho_0,\rho_1$ of $\PmvA$, and
  for all $\PmvB \subseteq \PmvA$, $\sysSi$, $\TxTS$,
  and $\TxY \subseteq \TxT$, if
  \[
  \sysS\rho_0 \xrightarrow{\TxTS} \sysSi
  \qquad \text{ with }
  \TxTS \in \mall{\rho_0(\PmvB)}{\TxY}^*
  \]
  then there exist $\TxTiS$ and $\sysRi$ such that:
  \[
  \sysS\rho_1 \xrightarrow{\TxTiS} \sysRi
  \; \text{with }
  \TxTiS \in \mall{\rho_1(\PmvB)}{\TxY}^*
  \text{ and }
  \wal{}{\sysRi} = \wal{}{\sysSi} \rho_0^{-1} \rho_1
  \]
\end{defn}

We establish below two basic properties of clusters.
First, any subset of a cluster is still a cluster;
second, if $\PmvA$ is a cluster in $\sysS$ then it is also
a cluster of any $\PmvA$-renaming of $\sysS$.

\begin{lem}
  \label{lem:cluster}
  Let $\PmvA$ be a cluster in $(\sysS,\TxT)$. Then:
  \begin{enumerate}[\rm(1)]
  \item \label{lem:cluster:subset}
    if $\PmvAi \subseteq \PmvA$ and $\TxTi \subseteq \TxT$,
    then $\PmvAi$ is a cluster in $(\sysS,\TxTi)$;
  \item \label{lem:cluster:rho} 
  for all renamings $\rho$ of $\PmvA$,
  $\PmvA$ is a cluster in $(\sysS\rho,\TxT)$.
  \end{enumerate}
\end{lem}
\begin{proof}
  For item~\ref{lem:cluster:subset}, 
  let $\PmvAi \subseteq \PmvA$,
  let $\TxTi \subseteq \TxT$, 
  and let $\rho_0,\rho_1$ be renamings of~$\PmvAi$. 
  For $i \in \setenum{0,1}$, let:
  \[
  \rho_i'
  \; = \;
  \lambda x. \textit{ if } x \in \PmvAi \textit{ then } \rho_i(x) \textit{ else } x
  \]
  Clearly, $\rho_i'$ is a renaming of $\PmvA$
  and $\sysS\rho_i = \sysS\rho_i'$ for $i \in \setenum{0,1}$.
  Let $\PmvB \subseteq \PmvAi$, $\TxY \subseteq \TxTi$, and $\sysSi$, $\TxTS$
  as in the hypotheses.
  Since $\PmvB \subseteq \PmvA$, $\TxY \subseteq \TxT$,
  and $\rho_i(\PmvB) = \rho_i'(\PmvB)$,
  the other premises remain true:
  hence, we can exploit the fact that $\PmvA$ is a cluster in $(\sysS,\TxT)$
  to prove that $\PmvAi$ is a cluster in $(\sysS,\TxTi)$.

  For item~\ref{lem:cluster:rho}, 	
  let $\rho,\rho_0,\rho_1$ be renamings of $\PmvA$.
  Note that also $\rho_i \circ \rho$ is a renaming of $\PmvA$ for $i \in \setenum{0,1}$.
  Let $\PmvB \subseteq \PmvA$,
  let $\TxY \subseteq \TxT$, and
  let $(\sysS\rho)\rho_0 \xrightarrow{\TxTS} \sysSi$,
  with $\TxTS \in \mall{\rho_0(\PmvB)}{\TxY}^*$.
  We have that $(\sysS\rho)\rho_0 = \sysS(\rho_0 \circ \rho)$ and,
  for $\PmvBi = \rho^{-1}(\PmvB)$, 
  $\TxTS \in \mall{(\rho_0 \circ \rho)(\PmvBi)}{\TxY}^*$.
  Let $\rho_0' = \rho_0 \circ \rho$ and $\rho_1' = \rho_1 \circ \rho$.
  Since $\PmvA$ is a cluster in $(\sysS,\TxT)$,
  then there exist $\TxTiS$ and $\sysRi$ such that:
  \begin{enumerate}
  \item $\sysS\rho_1' \xrightarrow{\TxTiS} \sysRi$
  \item $\TxTiS \in \mall{\rho_1'(\PmvBi)}{\TxY}^*$
  \item $\wal{}{\sysSi} (\rho_1' \circ \rho_0'^{-1})  = \wal{}{\sysRi}$.
  \end{enumerate}
  The thesis follows from:
  \begin{enumerate}
  \item $\sysS\rho_1' = \sysS (\rho_1 \circ \rho) = (\sysS \rho) \rho_1$
  \item $\rho_1'(\PmvBi) = (\rho_1 \circ \rho)(\rho^{-1}(\PmvB)) = \rho_1(\PmvB)$
  \item Since $\rho_0 = \rho_0' \circ \rho^{-1}$ and $\rho_1 = \rho_1' \circ \rho^{-1}$, then
    \begin{align*}
      \wal{}{\sysSi} (\rho_1 \circ \rho_0^{-1})
      & = \wal{}{\sysSi} ((\rho_1' \circ \rho^{-1}) \circ (\rho_0' \circ \rho^{-1})^{-1})
      \\
      & = \wal{}{\sysSi} ((\rho_1' \circ \rho^{-1}) \circ (\rho \circ \rho_0'^{-1}))
      \\
      & = \wal{}{\sysSi} (\rho_1' \circ \rho_0'^{-1})
      \\
      & = \wal{}{\sysRi}
      \tag*{\qedhere}
    \end{align*}
  \end{enumerate}  
\end{proof}

\begin{example}
  \label{ex:cluster:coinpusher:app}
  In the \txcode{CoinPusher} contract in~\Cref{fig:coinpusher},
  for any $\sysS$ and for
  $\TxT = \setenum{\txcode{play}(\txin{\pmvA[i]}{n_i}{\tokT})}_{i=1..n}$,
  we have that
  $\PmvC = \PmvU \setminus \setenum{\pmvA[i]}_{i=1..n}$ is a cluster in $(\sysS,\TxT)$.
  To see how a sequence $\TxTS$ is transformed into an equivalent $\TxTiS$,
  consider the case $n=1$, $\pmvA[1] = \pmvA$ and $n_1 = 99$.
  Let $\PmvU = \setenum{\pmvA[i]}_{i \in \Nat}$,
  let $\PmvB \subseteq \PmvC$, $\TxY \subseteq \TxT$,
  and let $\rho_0, \rho_1$ be renamings of $\PmvC$.
  Any $\TxTS \in \mall{\rho_0(\PmvB)}{\TxY}^*$
  is a sequence of transactions of the form
  \(
  \txcode{play}(\txin{\pmvB[i]}{v_i}{\tokT})
  \),
  with $\pmvB[i] \in \rho_0(\PmvB) \cup \setenum{\pmvA}$.
  We craft $\TxTiS \in \mall{\rho_1(\PmvB)}{\TxY}^*$ from $\TxTS$,
  by renaming with $\rho_0^{-1}\rho_1$ the actors in $\TxTS$.
  For instance, consider a state $\sysS$ where the contract has $0$ tokens,
  and:
  \[
  \sysS =
  \walu{\pmvA[0]}{99}{\tokT} \mid
  \walu{\pmvA[1]}{1}{\tokT} \mid
  \walu{\pmvA[2]}{0}{\tokT} \mid
  \walu{\pmvA[3]}{0}{\tokT} \mid \cdots
  \]
  Let
  \(
  \rho_0 = \setenum{\bind{\pmvA[0]}{\pmvA[1]}, \bind{\pmvA[1]}{\pmvA[2]}, \bind{\pmvA[2]}{\pmvA[0]}}
  \),
  let $\PmvB = \setenum{\pmvA[0],\pmvA[1]}$, and let
  \[
  \TxTS = \txcode{play}(\txin{\pmvA[1]}{99}{\tokT}) \; \txcode{play}(\txin{\pmvA[2]}{1}{\tokT})
  \in \mall{\rho_0(\PmvB)}{\TxY}
  \]
  Now, let
  \(
  \rho_1 = \setenum{\bind{\pmvA[0]}{\pmvA[3]}, \bind{\pmvA[1]}{\pmvA[0]}, \bind{\pmvA[3]}{\pmvA[1]}}
  \).
  We define $\TxTiS$ as:
  \[
  \TxTiS = \txcode{play}(\txin{\pmvA[3]}{99}{\tokT}) \; \txcode{play}(\txin{\pmvA[0]}{1}{\tokT})
  \in \mall{\rho_1(\PmvB)}{\TxY}
  \]
  Firing $\TxTS$ and $\TxTiS$, respectively,
  from $\sysS\rho_0$ and $\sysS\rho_1$ leads to:
  \begin{align*}
    \sysS\rho_0
    & \xrightarrow{\TxTS\;}
    \walu{\pmvA[0]}{0}{\tokT} \mid
    \walu{\pmvA[1]}{0}{\tokT} \mid
    \walu{\pmvA[2]}{100}{\tokT} \mid
    \walu{\pmvA[3]}{0}{\tokT} \mid \cdots    
    \\
    \sysS\rho_1
    & \xrightarrow{\TxTiS}
    \walu{\pmvA[0]}{100}{\tokT} \mid
    \walu{\pmvA[1]}{0}{\tokT} \mid
    \walu{\pmvA[2]}{0}{\tokT} \mid
    \walu{\pmvA[3]}{0}{\tokT} \mid \cdots        
  \end{align*}
  and the two states are equivalent up-to renaming.
  \hfill\qedex
\end{example}

\begin{figure}
  \begin{lstlisting}[language=txscript,morekeywords={DoubleAuth,init,auth1,auth2,withdraw},classoffset=3,morekeywords={a,b,A,B},keywordstyle=\pmvColor,classoffset=4,morekeywords={t,T},keywordstyle=\tokColor,frame=single]
contract DoubleAuth {
  init(a pays x:T) {
    require owner==null;
    owner=a;
  }  
  auth1(a sig) {
    require owner!=null && c1==null;
    require a==A || a==B || a==owner;
    c1=a;
  }  
  auth2(a sig) {
    require c1!=null && c2=null && c1!=a;
    require a==A || a==B || a==owner;
    c2=a;
  }  
  withdraw(a sig) {
    require c1!=null && c2!=null;
    transfer(a,balance(T):T);
  }
}
  \end{lstlisting}
  \negcaptionspace
  \caption{A contract requiring two authorizations.}
  \label{fig:doubleauth}
\end{figure}

\begin{example}
  \label{lem:doubleauth:cluster} 
  The contract \txcode{DoubleAuth} in~\Cref{fig:doubleauth}
  allows anyone to withdraw once 2 out of 3 actors
  among $\pmvA$, $\pmvB$ and the owner give their authorization.
  The contract has two hard-coded actors $\pmvA$ and $\pmvB$
  who have a privileged status (they can authorise the \txcode{withdraw}),
  and so cannot be replaced by others.
  Further, once \contract{DoubleAuth} has been initialised, also
  the owner (say, $\pmvO$) acquires the same privilege, and so 
  also $\pmvO$ is not replaceable by others.
  We prove that neither
  $\pmvA$ nor $\pmvB$ nor $\pmvO$ can belong to large enough clusters.

  Let $\sysS$ be the state reached upon firing
  $\txcode{init}(\txin{\pmvO}{1}{\tokT})$ from
  an initial state $\sysS[0]$ such that 
  $\wal{\pmvO}{\sysS[0]}= \idxfun{\tokT}$ and
  $\wal{\pmva}{\sysS[0]}(\tokt) = 0$ for all 
  $(\pmva,\tokt) \neq (\pmvO,\tokT)$.
  We prove that no cluster with size greater than 3 can include $\pmvA$, $\pmvB$ or $\pmvO$. 

  By contradiction, assume that $\PmvA$ is a cluster in $\sysS$ with $|\PmvA| > 3$ and $\pmvA \in \PmvA$.
  Let $\rho_0$ be the identity, and let 
  $\rho_1$ be such that $\rho_1(\pmvA) = \pmvM$,
  where $\pmvM \in \PmvA \setminus \setenum{\pmvA,\pmvB,\pmvO}$
  (such an $\pmvM$ always exists by the hypothesis on the size of $\PmvA$).
  Let $\PmvB = \setenum{\pmvA} \subseteq \PmvA$, and let
  $\TxTS = \txcode{auth1}(\txin{\pmvA}{0}{\tokT})$ $\txcode{auth2}(\txin{\pmvB}{0}{\tokT})$ $\txcode{withdraw}(\txin{\pmvA}{0}{\tokT})$
  be a sequence in $\mall{\PmvA}{\TxY}^*$
  for some $\TxY$ with $\auth{\TxY} \cap \PmvA = \emptyset$
  and $\pmvO \not\in \auth{\TxY}$
  (this can always be satisfied by choosing $\TxY = \emptyset$ if $\pmvB \in \PmvA$, and $\TxY = \setenum{\pmvB}$ otherwise).
  Since $\PmvA$ is a cluster in $\sysS$, there must exist
  $\TxTiS$ and $\sysRi$ satisfying~\Cref{def:cluster}. 
  Firing the \txcode{withdraw} requires $\auth{\TxY} \cup \rho_1(\PmvA)$ to contain at least two distinct actors among $\pmvA$, $\pmvB$ and $\pmvO$,
  while by the hypotheses above we have that
  $\auth{\TxY} \cup \rho_1(\PmvA) \subseteq \setenum{\pmvB,\pmvM}$.
  Therefore, $\TxTiS$ cannot contain a valid \txcode{withdraw},
  and so $\pmvM$ has no tokens in $\sysRi$.
  Instead, upon firing $\TxTS$ we have that 
  $\pmvA$ owns  $\waltok{1}{\tokT}$.
  This violates~\Cref{def:cluster},
  which requires $\wal{}{\sysRi} \pmvM = (\wal{}{\sysSi}\rho^{-1}) \pmvM = \wal{}{\sysSi} \rho(\pmvM) = \wal{}{\sysSi} \pmvA$.
  Proving that $\pmvB$ or $\pmvO$ cannot belong to clusters of size greater than 3 in $\sysS$ is done similarly.
  \hfill\qedex
\end{example}



\begin{example}
  \label{ex:cluster:coinpusher}
  Recall the contract \txcode{CoinPusher} in~\Cref{fig:coinpusher}.
  For any $\sysS$ and for
  $\TxT = \setenum{\txcode{play}(\txin{\pmvA[i]}{n_i}{\tokT})}_{i=1..n}$,
  we have that
  $\PmvC = \PmvU \setminus \setenum{\pmvA[i]}_{i=1..n}$
  is a cluster in $(\sysS,\TxT)$.
  \hfill\qedex
\end{example}

The following \namecref{th:mev:cluster-renaming} establishes that
renaming the actors in a cluster preserves their MEV.

\begin{thm}[Preservation under renaming]
  \label{th:mev:cluster-renaming}
  Let $\PmvA$ be a cluster in $(\sysS,\TxT)$
  and let $\rho$ be a renaming of $\PmvA$.
  For all~$\PmvB \subseteq \PmvA$:
  \[
  \mev{\PmvB}{\sysS}{\TxT} \; = \; \mev{\rho(\PmvB)}{\sysS \rho}{\TxT}
  \]
\end{thm}
\begin{proof}
  Since $\PmvA$ is a cluster in $(\sysS,\TxT)$, then
  by~\Cref{lem:cluster}\ref{lem:cluster:rho},
  $\PmvA$ is also a cluster in $(\sysS\rho,\TxT)$.
  Hence, by symmetry it is enough to prove $\leq$ between MEVs,
  which is implied by:
  \begin{align}
    & \label{eq:mev:cluster-renaming:1}
    \max \gain{\PmvB}{\sysS}{\mall{\PmvB}{\TxT}^*}
    \leq
    \max \gain{\rho(\PmvB)}{\sysS\rho}{\mall{\rho(\PmvB)}{\TxT}^*}
  \end{align}
  To prove~\eqref{eq:mev:cluster-renaming:1},
  let $\sysS \xrightarrow{\TxTS} \sysSi$, where
  $\TxTS \in \mall{\PmvB}{\TxT}^*$ maximises the gain of $\PmvB$.
  Since $\PmvB \subseteq \PmvA$, by~\Cref{def:cluster}
  (choosing the identity function for the first renaming $\rho_0$ and $\rho$ for the second renaming $\rho_1$),
  there exist $\TxTiS$ and $\sysRi$ such that
  $\sysS\rho \xrightarrow{\TxTiS} \sysRi$
  with $\TxTiS \in \mall{\rho(\PmvB)}{\TxT}^*$
  and $\wal{}{\sysSi} \rho = \wal{}{\sysRi}$.
  Then, $\wal{\rho(\PmvB)}{\sysRi} = \wal{\PmvB}{\sysSi}$.
  From this we obtain:
  \begin{align*}
    \max \gain{\PmvB}{\sysS}{\mall{\PmvB}{\TxT}^*}
    & = \gain{\PmvB}{\sysS}{\TxTS}
    \\
    & = \gain{\rho(\PmvB)}{\sysS\rho}{\TxTiS}
    \\
    & \leq \max \gain{\rho(\PmvB)}{\sysS\rho}{\mall{\rho(\PmvB)}{\TxT}^*}
    \tag*{\qedhere}
  \end{align*}
\end{proof}


%
%

The following~\namecref{lem:cluster-is-attacker} shows that
infinite clusters with a positive MEV are also MEV-attackers.
This is an important sanity check for our definition of MEV-attacker,
since actors in the same cluster can be renamed without affecting their MEV.

\begin{prop}
  \label{lem:cluster-is-attacker}
  Let $\PmvA$ be an infinite cluster of $(\sysS,\TxT)$ such that
  $\mev{\PmvA}{\sysS}{\TxT} > 0$.
  Then, $\PmvA$ is a MEV-attacker in $(\sysS,\TxT)$.
\end{prop}
\begin{proof}
  By~\Cref{th:mev:finite-part}, there exists some $\PmvAfin \subseteqfin \PmvA$
  such that $\mev{\PmvAfin}{\sysS}{\TxT} = \mev{\PmvB}{\sysS}{\TxT} = \mev{\PmvA}{\sysS}{\TxT} > 0$
  for all $\PmvAfin \subseteq \PmvB \subseteq \PmvA$.
  Let $\PmvB$ be an infinite subset of $\PmvA$.
  Take $\rho$ be any renaming of $\PmvA$ such that
  $\PmvAfin \subseteq \rho(\PmvB) \subseteq \PmvA$ and
  $\PmvA \setminus \rho(\PmvB)$ has no tokens in $\sysS$
  (note that such a renaming always exists, whether $\PmvB$ is cofinite or not).
  Then, $\mev{\rho(\PmvB)}{\sysS}{\TxT} > 0$.
  Since $\PmvA$ is a cluster and $\rho^{-1}$ is a renaming of $\PmvA$, 
  by~\Cref{th:mev:cluster-renaming} we have:
  \begin{align*}
    \mev{\PmvB}{\sysS \rho^{-1}}{\TxT}
    & = \mev{\rho^{-1}(\rho(\PmvB))}{\sysS \rho^{-1}}{\TxT}
    = \mev{\rho(\PmvB)}{\sysS}{\TxT} > 0
  \end{align*}
  To conclude, we note that
  $\sysS \rho^{-1}$ is a $(\PmvA,\PmvB)$-wallet redistribution of $\sysS$.
  Item 4 of~\Cref{def:wallet-redistribution} is trivial;
  for the other items, we must show that the tokens of $\PmvA$ in $\sysS$
  are tokens of $\PmvB$ in $\sysS \rho^{-1}$.
  In particular, it suffices to prove that
  $\PmvA \setminus \PmvB$ has no tokens in $\sysS \rho^{-1}$.
  By contradiction, let $\pmvA \in \PmvA \setminus \PmvB$ have tokens
  in $\sysS \rho^{-1}$.
  This implies that $\rho(\pmvA) \in \rho(\PmvA \setminus \PmvB) = \PmvA \setminus \rho(\PmvB)$
  must have tokens in $\sysS \rho^{-1} \rho = \sysS$.
  This contradicts the assumption that
  $\PmvA \setminus \rho(\PmvB)$ has no tokens in $\sysS$.
\end{proof}

\subsection{Universal MEV (\Cref{sec:adv-mev})}

\begin{proofof}{prop:adv-mev:monotonic-on-tx}
Consequence of~\Cref{prop:mev:monotonic-on-tx} and~\Cref{prop:badmev:monotonic-on-tx}.
\end{proofof}

\begin{proofof}{prop:adv-mev:monotonic-on-wal}
  It is sufficient to prove that
  for all $\PmvA$ and for all $\waldistr{}{\sysS}{}{\sysSi}$,
  there exists some $\WmvAi[\Delta]$ and token redistribution
  $\waldistr{}{\sysS+\WmvA[\Delta]}{}{\sysSi+\WmvAi[\Delta]}$
  such that
  $\mev{\PmvA}{\sysSi}{\TxT} \leq \mev{\PmvA}{\sysSi+\WmvAi[\Delta]}{\TxT}$.
  It is easy to transform the first token redistribution into the second one:
  it suffices to arbitrarily reassign the tokens in $\WmvA[\Delta]$.
  By~\Cref{prop:adv-mev:monotonic-on-wal},
  $\mev{\PmvA}{\sysSi}{\TxT} \leq \mev{\PmvA}{\sysSi+\WmvAi[\Delta]}{\TxT}$.
\end{proofof}

}
{}

\end{document}